\documentclass[11pt]{article}
\usepackage[margin=1.0in]{geometry}
\usepackage[T1]{fontenc}
\usepackage{microtype}
\usepackage{xspace}
\usepackage{amsmath}
\usepackage{amsfonts}
\usepackage{amssymb}
\usepackage{amsthm}
\usepackage{mathtools}
\usepackage{mathrsfs}
\usepackage{parskip}
\usepackage{changepage}
\usepackage{tcolorbox}
\usepackage{bbm}
\usepackage{tikz-cd}
\definecolor{darkblue}{RGB}{40,100,160}
\usepackage[colorlinks=true,
    linkcolor=darkblue,
    citecolor=darkblue,
    urlcolor=darkblue]{hyperref}
\usepackage[nameinlink]{cleveref}
\usepackage{aliascnt}
\usepackage{enumitem}
\usepackage{stmaryrd}

\usepackage{tabularx}
\usepackage{booktabs}

\numberwithin{equation}{section}

\theoremstyle{plain}
\newtheorem{theorem}{Theorem}[section]
\newtheorem*{theorem*}{Theorem}
\newaliascnt{lemma}{theorem}
\newtheorem{lemma}[lemma]{Lemma}
\aliascntresetthe{lemma}
\Crefname{lemma}{Lemma}{Lemmas}
\newaliascnt{corollary}{theorem}

\aliascntresetthe{corollary}
\Crefname{corollary}{Corollary}{Corollaries}
\newaliascnt{proposition}{theorem}
\newtheorem{proposition}[proposition]{Proposition}
\aliascntresetthe{proposition}
\Crefname{proposition}{Proposition}{Propositions}
\newaliascnt{claim}{theorem}

\aliascntresetthe{claim}
\Crefname{claim}{Claim}{Claims}
\newaliascnt{definition}{theorem}
\newtheorem{definition}[definition]{Definition}
\aliascntresetthe{definition}
\Crefname{definition}{Definition}{Definitions}
\theoremstyle{definition}
\newaliascnt{assumption}{theorem}

\aliascntresetthe{assumption}
\Crefname{assumption}{Assumption}{Assumptions}
\theoremstyle{remark}
\newaliascnt{remark}{theorem}
\newtheorem{remark}[remark]{Remark}
\aliascntresetthe{remark}
\Crefname{remark}{Remark}{Remarks}

\newcommand{\bbN}{\mathbb{N}}

\newcommand{\bbR}{\mathbb{R}}

\newcommand{\mcB}{\mathcal{B}}
\newcommand{\mcC}{\mathcal{C}}
\newcommand{\mcD}{\mathcal{D}}
\newcommand{\mcE}{\mathcal{E}}
\newcommand{\mcF}{\mathcal{F}}

\newcommand{\mcL}{\mathcal{L}}

\newcommand{\mcQ}{\mathcal{Q}}

\newcommand{\mcW}{\mathcal{W}}
\newcommand{\mcX}{\mathcal{X}}

\DeclarePairedDelimiterX{\cbrak}[2]{[}{]}{{#1}\,\delimsize|\,{#2}}

\DeclarePairedDelimiterX{\ip}[2]{\langle}{\rangle}{{#1},{#2}}

\DeclarePairedDelimiterX{\braket}[2]{\langle}{\rangle}{{#1}\delimsize\vert{#2}}
\DeclarePairedDelimiterX{\ketbra}[2]{\vert}{\vert}{{#1}\delimsize\rangle\!\delimsize\langle{#2}}
\DeclarePairedDelimiterX{\matrixel}[3]{\langle}{\rangle}{{#1}\delimsize\vert{#2}\delimsize\vert{#3}}

\newcommand{\F}{\mathbb F}
\newcommand{\E}{\mathbb E}
\newcommand{\1}{\mathbf 1}
\newcommand{\eps}{\varepsilon}
\newcommand{\wt}{\operatorname{wt}}

\newcommand{\Tr}{\operatorname{Tr}}

\newcommand{\Fold}{\operatorname{Fold}}

\newcommand{\poly}{\operatorname{poly}}
\newcommand{\calB}{\mathcal B}

\newcommand{\val}{\mathrm{val}}

\title{Explicit Capacity-Achieving Quantum LDPC Codes \\
        List Decodable in Near-linear Time}
\author{}

\author{William Gay\thanks{{\tt University of Illinois, Urbana-Champaign}. {\tt whgay2@illinois.edu}. }  
   \and Fernando Granha Jeronimo\thanks{{\tt University of Illinois, Urbana-Champaign}. {\tt granha@illinois.edu}. } 
   \and Abhi Shukul\thanks{{\tt University of Illinois, Urbana-Champaign}. {\tt ashukul2@illinois.edu}. }}

\date{\today}

\begin{document}
\emergencystretch=2em
\maketitle

\begin{abstract}
  In classical coding theory, the quest for explicit codes achieving list decoding capacity has been an important driving force. While random codes are easily shown to achieve capacity, an explicit construction was only discovered decades later in the seminal work [Guruswami and Rudra, STOC 2006]. In quantum coding theory, it is highly desirable that a code family be LDPC. While good quantum codes were known for decades, obtaining the additional LDPC property was elusive. In fact, only very recently that good quantum LDPC codes were discovered in a breakthrough work [Panteleev and Kalachev, STOC 2022]. In this context, a natural question is to ask for an explicit family of
  quantum codes achieving list decoding capacity while also possessing the important LDPC property.

 In this work, we provide (to the best of our knowledge) the first explicit constructions of quantum LDPC codes achieving list decoding capacity, namely, with a list decoding radius approaching the quantum Singleton bound with constant list sizes. Furthermore, we provide near-linear time (in the block-length) list decoding algorithms approaching capacity.

 Our explicit code construction are based on expander graphs via the quantum analogue of Alon--Edmonds--Luby (AEL) amplification [Bergamaschi, Golowich and Gunn, STOC 2024], and this enables the important LDPC property. Our efficient list decoding algorithms are obtained by generalizing the classical expander-based weak-regularity list decoding algorithms [Srivastava and Tulsiani, FOCS 2025] [Jeronimo and Singh, 2025] to suitable instantiations of quantum AEL.
\end{abstract}

\clearpage
\setcounter{tocdepth}{2}
\tableofcontents
\clearpage

\section{Introduction}\label{sec:introduction}

Error-correcting codes protect information by adding redundancy, and the central quantitative question of coding theory is how much protection a given amount of redundancy can buy. For a classical code of block length $n$, dimension $k$, and minimum distance $d$, the Singleton bound states that $d\le n-k+1$: writing $R=k/n$ for the rate and $\delta=d/n$ for the relative distance, no code can beat $\delta\le 1-R$. A decoder required to identify the transmitted codeword uniquely corrects fewer than $d/2$ errors, so unique decoding stops at the relative radius $(1-R)/2$ even for optimal codes. \emph{List decoding} \cite{Elias57,Wozencraft58} removes this ceiling: the decoder outputs a short list guaranteed to contain every codeword near the received word, and the correctable radius can then be pushed to the distance scale $\delta$. 

Up to the information-theoretic limit $1-R$, explicit capacity-achieving codes and efficient decoders were discovered in the seminal work of Guruswami and Rudra \cite{GuruswamiR06} via beautiful algebraic construction of folded-Reed Solomon codes. Despite many amazing list decoding properties~\cite{AGGLZ25}, several algebraic constructions are not LDPC. 
On the combinatorial side, we have the expander-based code constructions: Tanner codes \cite{Tanner81,SS96} and the Alon--Edmonds--Luby (AEL) distance amplification \cite{AEL95}. 
These expander-based constructions can lead to LDPC codes and also to efficient decoding \cite{GI01,GI02,GI03}. However, it was not known whether these combinatorial construction could match the list decoding properties of algebraic codes.

A complementary line of work has recently unveiled that these combinatorial expander-based constructions can also provide amazing list decoding properties while at the same time being LDPC and constant alphabet size. Initially, Sum-of-Squares list decoding algorithms were obtained up 
to the Johnson bound \cite{JST23}. In~\cite{JMST25}, it was shown that AEL can achieve near optimal list sizes analogous to folded-Reed Solomon codes from the breakthrough result of Chen and Zhang~\cite{CZ25}. In \cite{JeronimoSingh25,ST25}, near-linear time list decoding algorithms approaching capacity based on Frieze and Kannan \cite{FK96} weak regularity were shown for these expander-based constructions. Very recently, in \cite{JS26}, it was shown that AEL can be used to derandomize a broad class of properties (LCL) of linear codes \cite{LMS25}. Concurrently, \cite{BCDZ25} showed an analogous result for folded-Reed Solomon codes. These results reveal the beautiful complementarity of the combinatorial and algebraic sides of coding theory and how both can approach optimal trade-offs.

Quantum error correction inherits this landscape with two important changes. A quantum CSS code \cite{CalderbankShor96,Steane96} on $n$ qudits over $\F_q$ is specified by a pair of classical codes $C_X,C_Z\subseteq\F_q^n$ with $C_Z^\perp\subseteq C_X$; the space $C_Z^\perp$ (respectively $C_X^\perp$) consists of the $X$-type (respectively $Z$-type) \emph{stabilizers}, which act trivially on the encoded information. The first change is quantitative: the quantum Singleton bound \cite{KnillLaflamme97,Rains99Nonbinary} states that $k\le n-2(d-1)$, so the relative distance of a quantum code is capped at $\delta\le (1-R)/2+o(1)$, namely, half the classical bound, and unique decoding is correspondingly capped near the $(1-R)/4$ scale. The second change is structural: because stabilizers act trivially, an error is physically meaningful only \emph{modulo stabilizers}. Every sound notion of decoding for quantum codes is therefore a statement about equivalence classes: a unique decoder must return \emph{some} error in the correct class, and a list decoder must return a short list of classes. In the language of CSS codes, these are cosets $h_X+C_Z^\perp$ containing every class consistent with a nearby corruption. 

A quantum code is LDPC if every stabilizer check acts on $O(1)$ qudits and every qudit participates in $O(1)$ checks. This is the regime in which syndrome extraction is local and plays a crucial role in quantum fault tolerance. After a long search, explicit asymptotically good\footnote{Constant rate and linear distance.} quantum LDPC codes were constructed in the breakthrough work of Panteleev and Kalachev \cite{PK22} using expander 
graphs. Subsequently, an alternative explicit construction, the quantum Tanner codes, was given by Leverrier and Z\'emor \cite{LeverrierZemor22QuantumTanner}. Efficient decoders followed \cite{PK22} quickly:  Dinur--Hsieh--Lin--Vidick, Gu--Pattison--Tang, and Leverrier--Z\'emor 
gave decoders that correct a constant fraction of adversarial errors in linear time \cite{DinurHsiehLinVidick23, GuPattisonTang23, LeverrierZemor22DecodingQTanner}. All of these are \emph{unique} decoders: they operate below half the minimum distance, hence below the $(1-R)/4$ scale. 

Going much beyond the unique decoding radius, Bergamaschi, Golowich, and Gunn \cite{BergamaschiGolowichGunn22AQEC} showed that \emph{approximate} quantum error correction, a notion going back to Cr{\'e}peau, Gottesman, and Smith~\cite{CGS05} and list coding of Leung and Smith \cite{LeungSmith08}, can approach the capacity $(1-R)/2$. They ported the AEL distance-amplification scheme to CSS codes (we call the result qAEL codes) and recovered the encoded state approximately after a near-Singleton fraction of errors. Moreover, they obtained non-LDPC quantum codes list decodable up to capacity. More recently, Shapiro and Matthews \cite{MS26} constructed non-LDPC quantum codes list decodable up to capacity over constant size alphabets.

For \emph{exact} list decoding of LDPC codes, Bergamaschi, Jeronimo, Mittal, Srivastava, and Tulsiani \cite{BergamaschiJeronimoMittalSrivastavaTulsiani24ListDecodableQLDPC} showed that qAEL codes can be list decoded in polynomial time up to the \emph{Johnson radius} $J(\delta)=1-\sqrt{1-\delta}$ via Sum-of-Squares relaxations. The Johnson radius, however, is far from the capacity, e.g. for $\delta\approx 1/2$, $J(\delta)\approx 1-1/\sqrt 2\approx 0.293$, well short of the target radius near $1/2$. Moreover, the Sum-of-Squares route incurs a large polynomial running time\footnote{$N^{\textrm{poly}(1/\epsilon)}$ where $N$ is the block-length and $\epsilon$ is the slack to the Johnson bound.}. 

In this context, the first motivating question of our work is the following:

\begin{center}
   \emph{Can we construct explict quantum LDPC codes combinatorially list decodable up to capacity?}
\end{center}

If such families are shown, then our next question is algorithmic:

\begin{center}
   \emph{Can we efficiently list decode explicit quantum LDPC codes up to capacity?}
\end{center}

\subsection{Main result}

We answer the both questions affirmatively. In fact, we show, to the best of our knowledge, 
the first explicit construction of quantum LDPC codes approaching capacity. 
Furthermore, we provide near-linear time (in the block-length) list decoding algorithms up to 
capacity for these codes.

\begin{theorem}[Main theorem]\label{thm:main}
For every $R\in(0,1)$ and every $\zeta,\xi>0$, there exist constants $b,\ell,w\in\bbN$,
an infinite set of block lengths $\mathcal N\subseteq\bbN$, and an explicit family
\[
    \{\mcQ_N:N\in\mathcal N\}
\]
of $\F_2$-linear vector-space CSS codes on $N$ folded blocks of size $b$ with the following properties for every $N\in\mathcal N$.
\begin{enumerate}[label=(\roman*),leftmargin=2.7em]
    \item Each $\mcQ_N$ has rate $R_N\ge R$.  Its $X$- and $Z$-check matrices have row and column weights at most $w$ on the underlying binary coordinates.
    \item There is an explicit lower bound $\delta_N\le\delta(\mcQ_N)$ satisfying
    \[
        \delta_N\ge \frac{1-R_N}{2}-\zeta.
    \]
    \item For some radius $\tau_N\ge\delta_N-\xi$, the code is $(\tau_N,\ell)$-list decodable: for every $X$- or $Z$-syndrome, at most $\ell = O_{R,\zeta,\xi}(1)$ stabilizer cosets contain a representative of block weight at most $\tau_NN$.\footnote{The list size $\ell$ scales like $\exp(\exp(\poly(1/\xi)))$. We did not attempt to optimize in this work.}
    \item There is a randomized syndrome-input decoder running in time
    \[
        \widetilde O_{R,\zeta,\xi}(N),
    \]
    which, with probability $1-o(1)$, outputs a list of at most $\ell$ cosets containing every such low-weight error class.  Every output representative is verified to have the prescribed syndrome.
\end{enumerate}
\end{theorem}

We note that the guarantee is only that every nearby coset is contained in the output list, so extra syndrome-consistent cosets may also be returned.  The theorem is stated for both $X$- and $Z$-errors, but in this paper, we will only give the $X$-decoder, since the $Z$-decoder follows simply by swapping $X$ and $Z$ due to their symmetry.

\begin{remark}
  Instantiating the construction with suitable choices of quantum AEL leads to explicit qLDPC codes with optimal list sizes $O(1/\epsilon)$ \cite{JMS26}. More precisely, we obtain list sizes at the quantum generalized Singleton bound and also optimal list sizes for list recovery \cite{JM26,JMS26}.
\end{remark}

\section{Proof Strategy}

The construction is a quantum AEL code built from three objects: a high-rate outer CSS code $\mcD$, a constant-size inner CSS code $\mcC$, and a regular bipartite expander $G=(L,R,E)$.  An outer symbol is encoded by $\mcC$, the resulting local coordinates are placed on the edges of $G$, and the edges are folded into blocks indexed by $R$.  The qAEL distance theorem~\cite{BergamaschiGolowichGunn22AQEC,BergamaschiJeronimoMittalSrivastavaTulsiani24ListDecodableQLDPC} tells us that the folded qAEL code inherits the distance of the inner code, with some loss depending on the expansion.  Our contribution is an algorithm that list decodes up to this distance in near-linear time.

\paragraph{Decoding from the syndrome and local lifting.}
The decoder receives a syndrome rather than a corrupted word. If we only consider the $X$ case (the $Z$ case follows in the same way), what this means is that if $H_X$ is the parity check matrix for $C_X$, then if a codeword $c_X \in C_X$ is corrupted by an error $e$ to get $c_X + e$, the list decoder receives the syndrome $s = H_X(c_X + e)$. A standard approach to list decoding involves solving Gaussian elimination to get some $w$ such that $H_Xw=s$. However, even when the parity-check matrix is sparse, the elimination can create dense intermediate matrices, which would lead to an algorithm that no longer takes near-linear time. We instead exploit some additional structure of the qAEL checks, namely that it consists of constant-size inner checks around each $u\in L$ and sparse checks inherited from the outer code.  From the local part of the syndrome, the decoder solves one constant-size system per left vertex and obtains a lift $r_u$.  If $e^R$ is any physical error with the given syndrome and $e$ is its unfolded edge word, then
\[
    w:=e+r
\]
has $w_u\in\mcC_X$ at every left vertex. Thus, the unknown error has been converted, without a global solve, into an unknown collection of inner codewords whose folding is close to the known folded lift $r^R$.  The outer part of the syndrome is carried forward as an \emph{affine outer syndrome} and never converted into a global received word.

\paragraph{The agreement CSP.}
We assume that the inner code is classically list decodable, so for each $u\in L$, our algorithm lists the inner codewords within radius $\rho_{\rm in}\Delta$ of $r_u$.  The number of such codewords will be at most a constant $\ell_{\rm in}$.  Introduce a variable $A_u\in[\ell_{\rm in}]$ selecting one local list entry.  For an edge $(u,v)$, the constraint asks whether the selected inner word agrees with $r$ on the port corresponding to that edge. This is a left-oriented agreement CSP, where variables live only on $L$, while the right vertices identify groups of edge constraints.

Why does this CSP encode list decoding?  Fix a low-weight error $e^R$ and the corresponding word $w=e+r$.  Let $T\subseteq R$ be the right blocks on which $e^R$ is zero.  Since the error has small folded weight, $T$ is large.  Expansion implies that for all but a small fraction of left vertices, most incident edges enter $T$; call this good set $S\subseteq L$.  On every $u\in S$, the true inner block $w_u$ lies in the local list.  Selecting it gives an assignment that satisfies every constraint in the rectangle $E(S,T)$.  Therefore, finding a short family of assignments covering all such assignments is exactly the candidate-generation part of quantum list decoding.

\paragraph{Rigidity, weak regularity, and the quantum difference.}
A satisfying assignment is useful only if it is rigid.  Changing its label at a good vertex must violate many constraints into $T$.  In a classical code, this follows from the distance between two distinct local codewords.  Quantumly there are two cases.  Two local candidates may represent different logical cosets, in which case the quotient distance of $\mcC_X/\mcC_Z^\perp$ gives the required separation.  Or they may differ by a nonzero stabilizer in $\mcC_Z^\perp$, in which case quotient distance sees no difference at all, so a separate lower bound on the weight of nonzero inner stabilizers is needed.  

We then apply the efficient weak-regularity decomposition for expanding CSPs~\cite{JST21,Jer23,JeronimoSingh25} to the constant number of edge-indicator functions describing the agreement instance.  The decomposition partitions $L$ into a constant number of atoms, depending only on the decoding slack and the local list size.  Constraint values are approximately determined by the label frequencies inside these atoms.  As in~\cite{JeronimoSingh25, ST25},
a count-preserving relabeling argument shows that a fully satisfying rigid assignment must be close to an assignment that is constant on every atom.  Enumerating all atom-constant assignments therefore produces a blocklength-independent candidate family, simultaneously covering every low-weight error coset on the success event of the regularity algorithm.

\paragraph{Outer stitching modulo stabilizers.}
Each surviving assignment is projected from an inner codeword to its logical component in $\mcC_X/\mcC_Z^\perp$, producing a noisy outer word.  The candidate-generation bound places this word within the unique-decoding radius of the outer quantum Tanner code.  The outer decoder returns a word correct only up to an outer stabilizer, which is exactly sufficient: after re-encoding and folding, changing the outer word by an element of $\mcD_Z^\perp$ changes the qAEL word only by an element of $\mcF_Z^\perp$.  Hence the stitching stage recovers the correct global error coset even though no canonical representative is ever identified.

\paragraph{Final Instantiation} Finally, in order to instantiate our code, we choose a constant-size inner code approaching the quantum singleton bound, a high-rate quantum Tanner outer code with a near-linear syndrome decoder, and a sufficiently strong constant-degree expander.  The outer rate is chosen close enough to one that the final rate has the desired value, and the expander degree is chosen large enough to make both the distance loss and the candidate-generation error fit within the prescribed gaps.

\section{Preliminaries}\label{sec:preliminaries}
This section fixes notation and states the algebraic and combinatorial facts used by the decoder.  Proofs of short linear algebra and expander facts are deferred to \Cref{app:prelim-proofs}. Advanced theorems are quoted from the literature.

Throughout, $\widetilde O(\cdot)$ suppresses polylogarithmic factors in the growing block length.  A randomized statement that holds \emph{with high probability} has success probability $1-o(1)$ as that block length tends to infinity, with all target and local parameters held fixed.  Decoding times assume that the code, graph, and fixed local linear maps have already been constructed; all field sizes and local block sizes are treated as constants.

\subsection{CSS codes and Quantum List Decoding}\label{subsec:css-list}

Fix a prime power $q$.  For a positive integer $b$, a block symbol is an element of $\F_q^b$.  All linearity is over $\F_q$, even when $\F_q^b$ is identified set-theoretically with a larger field.  For $x,y\in(\F_q^b)^N$, define the block Hamming distance and weight by
\[
    \mathrm{dist}_b(x,y):=\big|\{i\in[N]:x_i\ne y_i\}\big|,
    \qquad
    \wt_b(x):=\mathrm{dist}_b(x,0).
\]

\begin{definition}\label{def:css-code}
A vector-space CSS code of blocklength $N$ and block size $b$ over $\F_q$ is a pair
\[
    \mcQ=(\mcQ_X,\mcQ_Z),
    \qquad
    \mcQ_X,\mcQ_Z\subseteq(\F_q^b)^N,
\]
of $\F_q$-linear subspaces satisfying
\[
    \mcQ_Z^\perp\subseteq\mcQ_X,
    \qquad\text{equivalently}\qquad
    \mcQ_X^\perp\subseteq\mcQ_Z,
\]
where orthogonality is taken with respect to the standard $\F_q$-bilinear inner product on $\F_q^{bN}$.
\end{definition}

All CSS codes used below have positive logical dimension.  The logical $\F_q$-dimension, folded dimension, and rate are
\[
    K_q(\mcQ):=\dim_{\F_q}\mcQ_X-\dim_{\F_q}\mcQ_Z^\perp,
    \qquad
    k(\mcQ):=\frac{K_q(\mcQ)}{b},
    \qquad
    R(\mcQ):=\frac{K_q(\mcQ)}{bN}.
\]
The $X$- and $Z$-distances are
\[
    d_X(\mcQ):=\min\{\wt_b(x):x\in\mcQ_X\setminus\mcQ_Z^\perp\},
    \qquad
    d_Z(\mcQ):=\min\{\wt_b(z):z\in\mcQ_Z\setminus\mcQ_X^\perp\},
\]
and the relative folded distance is
\[
    \delta(\mcQ):=\frac{\min\{d_X(\mcQ),d_Z(\mcQ)\}}{N}.
\]

Let $H_X:(\F_q^b)^N\to\F_q^{m_X}$ and $H_Z:(\F_q^b)^N\to\F_q^{m_Z}$ be parity-check maps with
\[
    \mcQ_X=\ker H_X,
    \qquad
    \mcQ_Z=\ker H_Z.
\]
A family is LDPC if both matrices have row and column weights bounded by constants independent of $N$, measured on the underlying $\F_q$-coordinates.

For $\tau\in(0,1)$ and $g\in(\F_q^b)^N$, write
\[
    B_\tau(g):=\{x\in(\F_q^b)^N:
                    \mathrm{dist}_b(x,g)\le\tau N\}.
\]
For syndromes $s_X\in\F_q^{m_X}$ and $s_Z\in\F_q^{m_Z}$, define the physical error lists
\begin{align*}
L_X(s_X,\tau)
    &:=\left\{
      e_X+\mcQ_Z^\perp \mid
      H_Xe_X=s_X,
      \ B_\tau(0)\cap (e_X+\mcQ_Z^\perp) \neq \emptyset
    \right\},\\
L_Z(s_Z,\tau)
    &:=\left\{
      e_Z+\mcQ_X^\perp \mid
      H_Ze_Z=s_Z,
      \ B_\tau(0)\cap (e_Z+\mcQ_X^\perp) \neq \emptyset
    \right\}.
\end{align*}

These are lists of cosets that contain a low-weight error vector that yields the given syndrome. Each element of the cosets yields the given syndrome because $\mcQ_Z^\perp \subseteq \ker H_X$ and $\mcQ_X^\perp \subseteq \ker H_Z$.

\begin{definition}\label{def:quantum-list-decoding}
For $\tau\in(0,1)$ and $\ell\in\bbN$, the code $\mcQ$ is \emph{$(\tau,\ell)$-list decodable} if
\[
    |L_X(s_X,\tau)|\le\ell,
    \qquad
    |L_Z(s_Z,\tau)|\le\ell
\]
for every pair of syndromes $s_X,s_Z$.
\end{definition}

\subsection{Bipartite Spectral Expanders}\label{subsec:expanders}

\begin{definition}\label{def:bipartite-expander}
Let $G=(L,R,E)$ be a $\Delta$-regular bipartite graph with $|L|=|R|=n$.  Let $A_G\in\bbR^{L\times R}$ be its bi-adjacency matrix, with singular values
\[
    \Delta=\sigma_1\ge\sigma_2\ge\cdots\ge\sigma_n\ge0.
\]
We call $G$ an $(n,\Delta,\lambda)$-bipartite expander if $\sigma_2\le\lambda$.
\end{definition}

For $S\subseteq L$ and $T\subseteq R$, write $E(S,T):=E\cap(S\times T)$.  The expander mixing lemma gives
\begin{equation}\label{eq:eml}
    \left||E(S,T)|-\frac{\Delta}{n}|S||T|\right|
    \le\lambda\sqrt{|S||T|}.
\end{equation}
We use the following standard consequence, whose proof appears in \Cref{app:prelim-proofs}.

\begin{lemma}\label{lem:neighborhood-concentration}
Fix $a,\beta\in(0,1)$ with $a>\beta$.  If $T\subseteq R$ satisfies $|T|\ge an$ and
\[
    S:=\{u\in L:|N(u)\cap T|\ge(a-\beta)\Delta\},
\]
then
\[
    |S|\ge\left(1-\frac{\lambda^2}{\Delta^2\beta^2}\right)n.
\]
The symmetric statement with $L$ and $R$ interchanged also holds.
\end{lemma}

\begin{theorem}[{Explicit Expanders, \cite[Thm.~1.3]{Alon21}}]\label{thm:alon-prescribed-size}
For every integer $\Delta\ge3$ and every $\eta>0$, all sufficiently large admissible $n$ admit a deterministic polynomial-time construction of a $\Delta$-regular graph on exactly $n$ vertices whose nontrivial adjacency eigenvalues have absolute value at most $2\sqrt{\Delta-1}+\eta$.  Taking the bipartite double cover yields an explicit $(n,\Delta,\lambda)$-bipartite expander with
\[
    \lambda\le2\sqrt{\Delta-1}+\eta.
\]
\end{theorem}

\subsection{Quantum AEL Codes}\label{subsec:qael}

Fix an $(n,\Delta,\lambda)$-bipartite expander $G=(L,R,E)$.  For each $x\in L\sqcup R$, fix a bijection
\[
    \mathrm{nbr}_x:N(x)\to[\Delta]
\]
such that, for every edge $e=(u,v)$,
\[
    \mathrm{idx}(e):=\mathrm{nbr}_u(v)=\mathrm{nbr}_v(u).
\]
Such a numbering follows from a proper $\Delta$-edge-coloring.
Fix an inner vector-space CSS code
\[
    \mcC=(\mcC_X,\mcC_Z),
    \qquad
    \mcC_X,\mcC_Z\subseteq(\F_q^{b_{\rm in}})^\Delta,
    \qquad
    \mcC_Z^\perp\subseteq\mcC_X,
\]
and an outer vector-space CSS code
\[
    \mcD=(\mcD_X,\mcD_Z),
    \qquad
    \mcD_X,\mcD_Z\subseteq(\F_q^{b_{\rm out}})^L,
    \qquad
    \mcD_Z^\perp\subseteq\mcD_X,
\]
with
\[
    b_{\rm out}=K_q(\mcC)
    =\dim_{\F_q}\mcC_X-\dim_{\F_q}\mcC_Z^\perp.
\]
Choose direct-sum decompositions
\begin{equation}\label{eq:WXWZ}
    \mcC_X=\mcW_X\oplus\mcC_Z^\perp,
    \qquad
    \mcC_Z=\mcW_Z\oplus\mcC_X^\perp,
\end{equation}
and let $\pi_X:\mcC_X\to\mcW_X$ and $\pi_Z:\mcC_Z\to\mcW_Z$ be the corresponding projections.

\begin{lemma}\label{lem:duality-embeddings}
There are $\F_q$-linear isomorphisms
\[
    \phi_X:\F_q^{b_{\rm out}}\to\mcW_X,
    \qquad
    \phi_Z:\F_q^{b_{\rm out}}\to\mcW_Z
\]
satisfying
\begin{equation}\label{eq:duality}
    \langle\phi_X(a),\phi_Z(b)\rangle=\langle a,b\rangle
    \qquad
    \text{for all }a,b\in\F_q^{b_{\rm out}}.
\end{equation}
Given bases of $\mcW_X$ and $\mcW_Z$, the maps can be found by Gaussian elimination over $\F_q$.
\end{lemma}

We can extend $\phi_X$ and $\phi_Z$ coordinatewise to get maps
\[
    \Phi_X,\Phi_Z:(\F_q^{b_{\rm out}})^L
    \longrightarrow
    \big((\F_q^{b_{\rm in}})^\Delta\big)^L.
\]
For an edge word $w\in((\F_q^{b_{\rm in}})^\Delta)^L$, define its folding $\Fold(w)\in((\F_q^{b_{\rm in}})^\Delta)^R$ by
\[
    \Fold(w)_{v,i}
    :=w_{\mathrm{nbr}_v^{-1}(i),i}
    \qquad
    (v\in R,\ i\in[\Delta]).
\]
Folding is a coordinate permutation, so it preserves the underlying inner product.

\begin{definition}\label{def:qael-code}
The concatenated CSS spaces induced by $\mcC$ and $\mcD$ are
\begin{align}
    \mcE_X
    &:=\Phi_X(\mcD_X)+(\mcC_Z^\perp)^L,
    \label{eq:EX}\\
    \mcE_Z
    &:=\Phi_Z(\mcD_Z)+(\mcC_X^\perp)^L.
    \label{eq:EZ}
\end{align}
The folded qAEL CSS code is $\mcF=(\mcF_X,\mcF_Z)$, where
\[
    \mcF_X:=\Fold(\mcE_X),
    \qquad
    \mcF_Z:=\Fold(\mcE_Z).
\]
It has folded blocklength $n$ and folded block size $\Delta b_{\rm in}$.
\end{definition}

With this definition of AEL amplification for quantum codes, the resulting codes are still CSS codes and they satisfy the same sort of properties as classical AEL-amplified codes.

\begin{proposition}[{\cite[Prop.~3.2]{BergamaschiJeronimoMittalSrivastavaTulsiani24ListDecodableQLDPC}}]\label{prop:dual-decomposition}
The dual spaces of the concatenated code are
\[
    \mcE_X^\perp
      =\Phi_Z(\mcD_X^\perp)+(\mcC_X^\perp)^L,
    \qquad
    \mcE_Z^\perp
      =\Phi_X(\mcD_Z^\perp)+(\mcC_Z^\perp)^L.
\]
Consequently,
\[
    \mcF_X^\perp=\Fold(\mcE_X^\perp),
    \qquad
    \mcF_Z^\perp=\Fold(\mcE_Z^\perp),
\]
and in particular $\mcF_Z^\perp\subseteq\mcF_X$.
\end{proposition}

\begin{proposition}\label{prop:qael-rate}
The folded rate satisfies
\[
    R(\mcF)=R(\mcC)R(\mcD).
\]
\end{proposition}

\begin{theorem}[{qAEL Distance Amplification, \cite[Thm.~4.2]{BergamaschiJeronimoMittalSrivastavaTulsiani24ListDecodableQLDPC}; see also \cite[Sec.~5]{BergamaschiGolowichGunn22AQEC}}]\label{thm:distance-amplification}
The relative folded distance of $\mcF$ satisfies
\[
    \delta(\mcF)
    \ge
    \delta(\mcC)-\frac{\lambda}{\Delta\,\delta(\mcD)}.
\]
We call any explicit lower bound obtained from the right-hand side a \emph{certified distance}.
\end{theorem}

We next give the parity-check representation used by the syndrome decoder.  Fix a basis $\{a_i:i\in I\}$ of $\mcC_X^\perp$ and a sparse spanning set $\{b_j:j\in J\}$ of $\mcD_X^\perp$, with $O(n)$ rows and $O(n)$ total nonzero base-field entries.  For $u\in L$, let $\iota_u(a_i)$ be the left edge word supported only at $u$, where it equals $a_i$.  Then a spanning set for $\mcF_X^\perp$ is
\begin{equation}\label{eq:qael-x-checks}
    \left\{\Fold(\iota_u(a_i)):u\in L,\ i\in I\right\}
    \ \cup\
    \left\{\Fold(\Phi_Z(b_j)):j\in J\right\}.
\end{equation}
The first family consists of local inner checks, and the second consists of outer checks.  Let
\[
    H_X:\big((\F_q^{b_{\rm in}})^\Delta\big)^R
      \longrightarrow
      \F_q^{L\times I}\oplus\F_q^J
\]
be the corresponding parity-check map.  For $H_Xe^R=(s^{\rm in},s^{\rm out})$, its entries are
\begin{align*}
    s^{\rm in}_{u,i}
       &=\langle\Fold(\iota_u(a_i)),e^R\rangle
         &&(u\in L,\ i\in I),\\
    s^{\rm out}_j
       &=\langle\Fold(\Phi_Z(b_j)),e^R\rangle
         &&(j\in J).
\end{align*}

\begin{proposition}\label{prop:qael-ldpc}
If the outer family $\mcD$ is LDPC and the inner parameters $q,b_{\rm in},b_{\rm out},\Delta$ are fixed constants, then the qAEL family $\mcF$ is LDPC on the underlying $\F_q$-coordinates.
\end{proposition}

The $X$-decoder uses two classical parameters of the inner space $\mcC_X$:
\begin{align}
    \delta_{\rm quot}(\mcC_X)
    &:=\frac1\Delta
      \min\{\wt_{b_{\rm in}}(c):c\in\mcC_X\setminus\mcC_Z^\perp\},
    \label{eq:delta-quot}\\
    \delta_{\rm stab}(\mcC_X)
    &:=\frac1\Delta
      \min\{\wt_{b_{\rm in}}(s):0\ne s\in\mcC_Z^\perp\}.
    \label{eq:delta-stab}
\end{align}
The $Z$-side parameters are defined symmetrically.  A minimum over an empty stabilizer set is taken to be $+\infty$.
\subsection{Constraint Satisfaction Problems}

When using regularity for decoding, it will be convenient to use the language of constraint satisfaction problems.

\begin{definition}
    A bipartite 2-CSP on $G = (L,R,E)$ with domain $[\ell]$ is a collection of constraints $\{\psi_{(u,v)}\subseteq [\ell]^2 \mid (u,v) \in E\}$. An assignment $a : L \cup R \to [\ell]$ is said to satisfy an edge $(u,v)$ when $(a(u),a(v)) \in \psi_{(u,v)}$. The value of an assignment for a given bipartite 2-CSP instance is the number of edges satisfied.
\end{definition}

We use the special case $\psi_{(u,v)}=\Psi_{(u,v)}\times[\ell]$, where $\Psi_{(u,v)}\subseteq[\ell]$.  These constraints depend only on the left label, so an assignment may be written as $a:L\to[\ell]$; the right vertices group the edge constraints.  We call this a \emph{left-oriented CSP}.

\subsection{Weak Regularity on Expanders}\label{subsec:weak-regularity}

For functions $f,g:\mcX \to \bbR$, we define the inner product 
$\langle f,g\rangle = \sum_{x\in \mcX} f(x)g(x)$. We also define $\langle f,g\rangle_W = \sum_{x\in W}f(x)g(x)$ for $W\subseteq \mcX$.

For $A\subseteq L$ and $B\subseteq R$, the function $\1_A\otimes\1_B$ is the indicator of the rectangle $A\times B$. For a function $f:L\times R\to\bbR$, we define the cut norm as
\[
    \|f\|_\square
    :=\max_{A\subseteq L,\,B\subseteq R}
        \big|\langle f,\1_A\otimes\1_B\rangle\big|.
\]

\begin{theorem}[{Efficient Weak Regularity on Expanders, \cite[Thm.~2.11]{JeronimoSingh25}; see also \cite{JST21,Jer23}}]\label{thm:weak-reg-expanded}
Let $G=(L,R,E)$ be an $(n,\Delta,\lambda)$-bipartite expander, and let $g:L\times R\to\{0,1\}$ be supported on $E$.  If
\[
    \frac{\lambda}{\Delta}<\frac{\gamma^2}{2^{23}},
\]
then a randomized algorithm outputs, with high probability, a function
\[
    h=\sum_{j=1}^{p}c_j\,\1_{A_j}\otimes\1_{B_j},
    \qquad
    p=O(\gamma^{-2}),
\]
such that
\[
    \|g-h\|_\square\le\gamma n\Delta.
\]
Its running time is
\[
    \widetilde O\!\left(2^{2^{O(\gamma^{-2})}}n\Delta\right).
\]
\end{theorem}

\begin{definition}\label{def:factors}
A factor $\calB$ on a finite set $\mcX$ is a partition of $\mcX$; its parts are called atoms.  For $W\subseteq\mcX$, the restricted factor is
\[
    \calB|_W:=\{P\cap W:P\in\calB,\ P\cap W\ne\emptyset\}.
\]
A function on $\mcX$ is $\calB$-measurable if it is constant on every atom.  A function $f:W\to\Sigma$ is $\eta$-concentrated on $\calB|_W$ if it agrees with some $\calB|_W$-measurable function on at least $(1-\eta)|W|$ points of $W$.
\end{definition}

A collection of cuts $A_1,\ldots,A_p\subseteq L$ generates the factor whose atoms are the nonempty membership patterns
\[
    \{u\in L:\1_{A_1}(u)=\alpha_1,\ldots,\1_{A_p}(u)=\alpha_p\},
    \qquad
    (\alpha_1,\ldots,\alpha_p)\in\{0,1\}^p.
\]
It has at most $2^p$ atoms. For cuts $B_1,\dots, B_p \subseteq R$, we get an analogous factor on $R$.

For a factor $\calB$ on $\mcX$, let $\calB(x)$ denote the atom containing $x$.  For $f:\mcX\to\bbR$, define
\[
    \E[f\mid\calB](x)
      :=\frac{1}{|\calB(x)|}\sum_{y\in\calB(x)}f(y).
\]

For any $\mcB$-measurable function $h$ on $\mcX$, it follows that 
\[
\langle h,f\rangle = \langle h,\mathbb{E}[f\mid \mcB]\rangle.
\]

\section{The List Decoding Algorithm}\label{sec:algorithm}
We present the syndrome-input list decoder and prove its correctness for $X$-errors.  The $Z$-error decoder follows by exchanging $X$ and $Z$ throughout.

\subsection{The Main Statement}\label{subsec:decoder-main}

Let $G = (L\cup R,V)$ be a bipartite graph with expansion $\lambda$ and $|L| = |R| = n$. We fix the following CSS codes: an inner code $(\mcC_X, \mcC_Z)$ in $(\F_q^{b_{in}})^\Delta$ and an outer code $(\mcD_X,\mcD_Z)$ in 
$(\F_q^{b_{out}})^n$. We let $(\mcF_X,\mcF_Z)$ be the corresponding folded qAEL CSS code with respect to $G$.

Fix a distance lower bound
\[
    0<\delta_R\le\delta(\mcF)
\]
and a slack $0<\eps<\delta_R$.  Define
\begin{equation}\label{eq:decoder-parameters}
    \tau:=\delta_R-\eps,
    \qquad
    \rho_{\rm in}:=\delta_R-\frac\eps2,
    \qquad
    \gamma:=\frac{\eps^3}{2^5\ell_{\rm in}}.
\end{equation}
We assume that $\mcC_X$ is classically $(\rho_{\rm in},\ell_{\rm in})$-list decodable in the block metric and that
\begin{equation}\label{eq:inner-threshold-hypotheses}
    \delta_{\rm quot}(\mcC_X)\ge\delta_R,
    \qquad
    \delta_{\rm stab}(\mcC_X)\ge\delta_R.
\end{equation}
This means there is an inner decoder $\mathsf{LD}_{\rm in}$ that returns, for every local word, all inner codewords within radius $\rho_{\rm in}\Delta$, using a list of size at most $\ell_{\rm in}$.  In the final construction all inner parameters are constant, so an exhaustive local search is sufficient.

Assume also that the outer code $\mcD$ has a syndrome decoder $\mathsf{Dec}_{\rm out}$ that corrects block weight at most $\rho_{\rm out}n$ up to $Z$-stabilizers in time $T_{\rm out}(n)$.  Set
\begin{equation}\label{eq:alpha-eta}
    \alpha:=\frac{4\lambda^2}{\Delta^2\eps^2},
    \qquad
    \eta_{\rm conc}:=\frac{5\eps^2}{2^5(1-\alpha)}.
\end{equation}
The spectral condition below implies $\alpha<1$.

\begin{theorem}[Syndrome-Input Decoder]\label{thm:syndrome-input-decoder}
Suppose
\begin{equation}\label{eq:spectral-condition-decoder}
    \frac{\lambda}{\Delta}
    <\frac{\gamma^2}{2^{23}}
    =\frac{\eps^6}{2^{33}\ell_{\rm in}^2}
\end{equation}
and
\begin{equation}\label{eq:stitching-inequality}
    \alpha+\eta_{\rm conc}\le\rho_{\rm out}.
\end{equation}
Then there is a randomized algorithm
$\mathsf{QAEL\mbox{-}Syndrome\mbox{-}Decode}_X$ such that, for every syndrome $s$,
\[
    L_X(s,\delta_R-\eps)
    \subseteq
    \mathsf{QAEL\mbox{-}Syndrome\mbox{-}Decode}_X(s)
\]
with probability $1-o_{n\to\infty}(1)$.  Every output representative has syndrome $s$, and the output contains at most
\begin{equation}\label{eq:Lmax-definition}
    L_{\max}
    :=\ell_{\rm in}^{\,2^{O(\ell_{\rm in}^3/\eps^6)}}
\end{equation}
cosets.  Its running time is
\[
    \widetilde O_{\eps,\ell_{\rm in}}(n)
    +L_{\max}\big(T_{\rm out}(n)+O(n)\big).
\]
In particular, for fixed local parameters and a near-linear outer decoder, the syndrome-input qAEL decoder is near-linear in $n$.
\end{theorem}

\subsection{Decoder Notation and Interfaces}\label{subsec:syndrome-lifting}

We first define the objects used by the algorithm.  Fix a syndrome
$s=(s^{\rm in},s^{\rm out})\in\F_q^{L\times I}\oplus\F_q^J$
for the check presentation in \eqref{eq:qael-x-checks}.  Thus $s^{\rm in}_{u,i}$ corresponds to the local check $\Fold(\iota_u(a_i))$, and $s^{\rm out}_j$ corresponds to the lifted outer check $\Fold(\Phi_Z(b_j))$.

\paragraph{The local lift $r$.}
Fix a complement $U_X$ with
\[
    (\F_q^{b_{\rm in}})^\Delta=\mcC_X\oplus U_X.
\]
The linear map
\[
    U_X\longrightarrow(\mcC_X^\perp)^*,
    \qquad
    z\longmapsto\bigl(x\mapsto\langle z,x\rangle\bigr)
\]
is an isomorphism, which means that since the $a_i$'s form a basis of $\mcC_X^\perp$, there is a unique $r_u\in U_X$ satisfying
\begin{equation}\label{eq:ru-definition}
    \langle r_u,a_i\rangle=-s^{\rm in}_{u,i}
    \qquad(i\in I).
\end{equation}
Set
\[
    r:=(r_u)_{u\in L},
    \qquad
    r^R:=\Fold(r).
\]
The inverse local map is precomputed, so constructing $r$ and $r^R$ takes $O(n)$ time. 

\paragraph{The affine outer syndrome and decoder.}
Define $\sigma_s\in\F_q^J$ by
\begin{equation}\label{eq:affine-sigma-definition}
    (\sigma_s)_j:=s^{\rm out}_j+\langle r,\Phi_Z(b_j)\rangle
    \qquad(j\in J).
\end{equation}
For any $\sigma\in\F_q^J$, let
\begin{equation}\label{eq:affine-outer-space}
    \mcD_X(\sigma)
    :=\{y\in(\F_q^{b_{\rm out}})^L:
         \langle y,b_j\rangle=\sigma_j\text{ for every }j\in J\}.
\end{equation}
This is a translate of $\mcD_X$ when nonempty; it may be empty if the syndrome is inconsistent with dependencies among the checks.  Computing $\sigma_s$ and checking membership in $\mcD_X(\sigma)$ each take $O(n)$ time using the sparse outer checks.

The procedure $\mathsf{AffDec}_{\rm out}(\widehat y,\sigma)$ takes a tentative outer word $\widehat y\in(\F_q^{b_{\rm out}})^L$ and a target syndrome $\sigma$.  It computes
\[
    t_j:=\langle\widehat y,b_j\rangle-\sigma_j
    \qquad(j\in J)
\]
and runs $\mathsf{Dec}_{\rm out}$ on $t$.  If that call returns a correction $c$, it forms $y^\star:=\widehat y-c$ and returns $y^\star$ only after verifying $y^\star\in\mcD_X(\sigma)$, and  otherwise it returns failure.  This procedure always returns a verified word or failure in time $T_{\rm out}(n)+O(n)$.  Its decoding guarantee is proved in \Cref{lem:affine-outer-decoder}.

\paragraph{Local lists, assignments, and projected words.}
For each $u\in L$, let
\begin{equation}\label{eq:local-list-definition}
    \mcL_u:=\{c\in\mcC_X:
        \mathrm{dist}_{b_{\rm in}}(c,r_u)\le\rho_{\rm in}\Delta\}.
\end{equation}
The inner list decoder computes $\mcL_u$, and $|\mcL_u|\le\ell_{\rm in}$.  Fix an ordering of its distinct codewords and let
\[
    \mathsf L_u:[\ell_{\rm in}]\longrightarrow\mcC_X\cup\{\bot\}
\]
list each codeword once, mapping every unused index to a dummy symbol $\bot$.  

An assignment $a:L\to[\ell_{\rm in}]$ selects one label at each left vertex.  Its \emph{projected left word} $w_a\in\mcW_X^L$ is
\begin{equation}\label{eq:wa-definition}
    w_a(u):=
    \begin{cases}
      \pi_X(\mathsf L_u(a(u))), & \mathsf L_u(a(u))\ne\bot,\\
      0, & \mathsf L_u(a(u))=\bot.
    \end{cases}
\end{equation}
The corresponding tentative outer word $\widehat y_a\in(\F_q^{b_{\rm out}})^L$ is defined by
\begin{equation}\label{eq:yhat-induced}
    (\widehat y_a)_u:=\phi_X^{-1}(w_a(u))
    \qquad(u\in L).
\end{equation}
In particular, $w_a=\Phi_X(\widehat y_a)$.  Neither $w_a$ nor $\widehat y_a$ is required to satisfy the outer checks at this stage.

\paragraph{The candidate-generation interface.} 
We write $\mathsf{CandGen}(r^R)$ for the randomized procedure that computes the local lists and returns a family $\mathcal A$ of at most $L_{\max}$ assignments, together with their projected words $(w_a)_{a\in\mathcal A}$.  The graph, inner code, and decoding parameters are fixed inputs to this procedure and its output depends on $r^R$ and the local lists.  \Cref{prop:candidate-generation} states the agreement guarantee needed for correctness while \Cref{subsec:csp-candidates} constructs the procedure and proves that guarantee using weak regularity.

\pagebreak[3]
\begin{samepage}
\subsection{The List Decoding Algorithm}\label{subsec:decoder-algorithm}

With these interfaces in place, the decoder computes the lift and affine syndrome, enumerates candidate left words, and applies the affine outer decoder to each one.

\begin{center}
\fbox{\begin{minipage}{0.94\linewidth}
\textbf{Algorithm $\mathsf{QAEL\mbox{-}Syndrome\mbox{-}Decode}_X$}

\textbf{Input:} a syndrome
$s=(s^{\rm in},s^{\rm out})\in\F_q^{L\times I}\oplus\F_q^J$.
\begin{enumerate}[leftmargin=1.6em]
    \item Compute the local lift $r=(r_u)_{u\in L}$ from \eqref{eq:ru-definition}, and set $r^R:=\Fold(r)$.
    \item Compute the affine outer syndrome $\sigma_s$ from \eqref{eq:affine-sigma-definition}.
    \item Run $\mathsf{CandGen}(r^R)$, obtaining at most $L_{\max}$ assignments $a\in\mathcal A$ and their projected left words $w_a:L\to\mcW_X$.
    \item For each $a\in\mathcal A$, form $\widehat y_a$ using \eqref{eq:yhat-induced} and run $\mathsf{AffDec}_{\rm out}(\widehat y_a,\sigma_s)$.
    \item Discard failures.  For every verified affine output $y^\star$, form
    \begin{equation}\label{eq:output-error-representative}
        e_a^\star:=\Fold(\Phi_X(y^\star)-r).
    \end{equation}
    Verify $H_Xe_a^\star=s$ and, if the check passes, append $e_a^\star$ as a representative of $e_a^\star+\mcF_Z^\perp$.
\end{enumerate}
\end{minipage}}
\end{center}
\end{samepage}

The output is stored as a list of representatives and so there may be multiple representatives for the same coset.  Its induced set of cosets is the final output in \Cref{thm:syndrome-input-decoder}.  This algorithm also does not test whether a coset contains a low-weight representative, so the output list may contain extra cosets. 

\subsection{Correctness}\label{subsec:decoder-correctness}

Fix an input syndrome $s$.  Our goal is to cover every coset in $L_X(s,\tau)$.  Equivalently, for every $e^R$ with
\[
    H_Xe^R=s,
    \qquad
    \wt_{\Delta b_{\rm in}}(e^R)\le\tau n,
\]
we must show that some enumerated assignment produces an output $e_a^\star$ satisfying
$e_a^\star-e^R\in\mcF_Z^\perp$.  The lemmas below reduce this goal to the candidate-generation guarantee in \Cref{prop:candidate-generation}, whose proof is deferred to \Cref{subsec:csp-candidates}.

\begin{lemma}\label{lem:local-lift-property}
Let $H_Xe^R=s$, write $e:=\Fold^{-1}(e^R)$, and set $w:=e+r$.  Then $w_u\in\mcC_X$ for every $u\in L$, and
\[
    \mathrm{dist}_{\Delta b_{\rm in}}(\Fold(w),r^R)
    =\wt_{\Delta b_{\rm in}}(e^R).
\]
\end{lemma}

\begin{proof}
For every $u\in L$ and $i\in I$,
\[
    \langle w_u,a_i\rangle
    =\langle e_u,a_i\rangle+\langle r_u,a_i\rangle
    =s^{\rm in}_{u,i}-s^{\rm in}_{u,i}=0.
\]
Since the $a_i$ span $\mcC_X^\perp$, we have $w_u\in\mcC_X$.  The distance identity follows from $\Fold(w)-r^R=\Fold(e)=e^R$.
\end{proof}

The shifted word $w$ is locally in $\mcC_X$, but need not lie in the concatenated code $\mcE_X$.  Its logical components instead satisfy the affine outer syndrome.

\begin{lemma}\label{lem:outer-affine-consistency}
Let $e^R,e,w$ be as in \Cref{lem:local-lift-property}.  Define $y\in(\F_q^{b_{\rm out}})^L$ by
\begin{equation}\label{eq:outer-word-from-w}
    y_u:=\phi_X^{-1}(\pi_X(w_u))
    \qquad(u\in L).
\end{equation}
Then $y\in\mcD_X(\sigma_s)$ and $w=\Phi_X(y)+g$ for some $g\in(\mcC_Z^\perp)^L$.
\end{lemma}

\begin{proof}
The direct sum $\mcC_X=\mcW_X\oplus\mcC_Z^\perp$ gives
$w=\Phi_X(y)+g$ with $g\in(\mcC_Z^\perp)^L$.  For every $j\in J$,
\begin{align*}
    \langle y,b_j\rangle
    &=\langle\Phi_X(y),\Phi_Z(b_j)\rangle\\
    &=\langle w,\Phi_Z(b_j)\rangle\\
    &=s^{\rm out}_j+\langle r,\Phi_Z(b_j)\rangle
      =(\sigma_s)_j.
\end{align*}
The first equality uses the duality of the embeddings.  The second uses
$\Phi_Z(b_j)\in\mcC_Z^L$, which is orthogonal to $g$; the third uses the lifted outer checks and $w=e+r$.
\end{proof}

\begin{lemma}[Affine outer decoding]\label{lem:affine-outer-decoder}
If $y\in\mcD_X(\sigma)$ and
$\mathrm{dist}_{b_{\rm out}}(\widehat y,y)\le\rho_{\rm out}n$, then
$\mathsf{AffDec}_{\rm out}(\widehat y,\sigma)$ returns a verified word $y^\star\in\mcD_X(\sigma)$ satisfying
\[
    y^\star-y\in\mcD_Z^\perp.
\]
Its running time is $T_{\rm out}(n)+O(n)$, including on inputs outside this promise.
\end{lemma}

\begin{proof}
Put $z:=\widehat y-y$.  The difference syndrome computed by the procedure is
\[
    t_j=\langle\widehat y,b_j\rangle-\sigma_j
       =\langle z,b_j\rangle.
\]
Since $\wt_{b_{\rm out}}(z)\le\rho_{\rm out}n$, the outer decoder returns a correction $c$ with $c-z\in\mcD_Z^\perp$.  Consequently
\[
    y^\star-y=\widehat y-c-y=z-c\in\mcD_Z^\perp.
\]
The inclusion $\mcD_Z^\perp\subseteq\mcD_X$ implies $y^\star\in\mcD_X(\sigma)$, so verification succeeds.  The time bound follows from the cap on the outer call and the $O(n)$ cost of computing and verifying the sparse syndromes.
\end{proof}

We now state the only property of candidate generation used in this subsection.

\begin{proposition}[Candidate-generation guarantee]\label{prop:candidate-generation}
Under the inner-code hypotheses above and the spectral condition \eqref{eq:spectral-condition-decoder}, the procedure $\mathsf{CandGen}(r^R)$ returns a family $\mathcal A$ of at most $L_{\max}$ assignments and their projected words.  With probability $1-o(1)$, the following holds simultaneously for every $e^R$ with
\[
    H_Xe^R=s,
    \qquad
    \wt_{\Delta b_{\rm in}}(e^R)\le\tau n.
\]
Writing $e=\Fold^{-1}(e^R)$ and $w=e+r$, some $a\in\mathcal A$ satisfies
\begin{equation}\label{eq:candidate-error-bound}
    \big|\{u\in L:w_a(u)\ne\pi_X(w_u)\}\big|
    \le(\alpha+\eta_{\rm conc})n.
\end{equation}
For fixed local parameters, the procedure runs in $\widetilde O_{\eps,\ell_{\rm in}}(n)$ time, including the explicit writing of all candidates.
\end{proposition}

The proposition is proved in \Cref{subsec:csp-candidates}.  Its simultaneous guarantee is important: one successful invocation must cover all low-weight error classes for the given syndrome.

\begin{lemma}[Stitching the correct coset]\label{lem:stitching}
Let $e^R$ have syndrome $s$, and define $w$ and $y$ as above.  If an assignment $a$ satisfies
\[
    \big|\{u\in L:w_a(u)\ne\pi_X(w_u)\}\big|
    \le\rho_{\rm out}n,
\]
then $\mathsf{AffDec}_{\rm out}(\widehat y_a,\sigma_s)$ succeeds, and its output $y^\star$ yields a representative $e_a^\star$ from \eqref{eq:output-error-representative} with
\[
    e_a^\star-e^R\in\mcF_Z^\perp.
\]
\end{lemma}

\begin{proof}
By \eqref{eq:yhat-induced}, \eqref{eq:outer-word-from-w}, and injectivity of $\phi_X$,
\[
    \mathrm{dist}_{b_{\rm out}}(\widehat y_a,y)
    =\big|\{u\in L:w_a(u)\ne\pi_X(w_u)\}\big|
    \le\rho_{\rm out}n.
\]
By \Cref{lem:outer-affine-consistency}, $y\in\mcD_X(\sigma_s)$.  Hence \Cref{lem:affine-outer-decoder} guarantees a verified output with $y^\star-y\in\mcD_Z^\perp$.

Write $w=\Phi_X(y)+g$ with $g\in(\mcC_Z^\perp)^L$.  Since $w=e+r$,
\[
    \Fold^{-1}(e_a^\star)-e
    =\Phi_X(y^\star)-r-e
    =\Phi_X(y^\star-y)-g.
\]
By \Cref{prop:dual-decomposition}, this difference belongs to
\[
    \Phi_X(\mcD_Z^\perp)+(\mcC_Z^\perp)^L=\mcE_Z^\perp.
\]
Folding gives $e_a^\star-e^R\in\mcF_Z^\perp$.
\end{proof}

\begin{proof}[Proof of \Cref{thm:syndrome-input-decoder}]
Assume the guarantee of \Cref{prop:candidate-generation}, proved below.

\emph{Coverage.}
Condition on its success event.  Fix any coset in $L_X(s,\tau)$, and choose a representative $e^R$ with syndrome $s$ and folded weight at most $\tau n$.  Set $w=\Fold^{-1}(e^R)+r$.  The proposition supplies an enumerated assignment $a$ with
\[
    \big|\{u\in L:w_a(u)\ne\pi_X(w_u)\}\big|
    \le(\alpha+\eta_{\rm conc})n
    \le\rho_{\rm out}n,
\]
where the last inequality is \eqref{eq:stitching-inequality}.  By \Cref{lem:stitching}, the corresponding affine outer call succeeds and produces $e_a^\star$ in the same stabilizer coset as $e^R$.  In particular, $H_Xe_a^\star=s$, so the final check does not discard it.  Because the success event is simultaneous over all target errors, every coset in $L_X(s,\tau)$ is covered in this one run.

\emph{Syndrome consistency.}
Every retained output is explicitly verified.  The construction also shows why its syndrome must be $s$: for any $y^\star\in\mcD_X(\sigma_s)$ and $e^\star=\Fold(\Phi_X(y^\star)-r)$,
\begin{align*}
    \langle e^\star,\Fold(\iota_u(a_i))\rangle
    &=\langle\phi_X(y_u^\star),a_i\rangle-\langle r_u,a_i\rangle
      =s^{\rm in}_{u,i},\\
    \langle e^\star,\Fold(\Phi_Z(b_j))\rangle
    &=\langle y^\star,b_j\rangle-\langle r,\Phi_Z(b_j)\rangle
      =s^{\rm out}_j.
\end{align*}
Here $\phi_X(y_u^\star)\in\mcC_X$ is orthogonal to $a_i$, and the second equality uses \eqref{eq:affine-sigma-definition}.  This conclusion holds for all outputs, not just those arising from the coverage argument.

\emph{List size and running time.}
Each of at most $L_{\max}$ assignments contributes at most one representative.  This bounds the output length even if cosets repeat, and also bounds the number of target cosets because any successful run covers them all.  Local lifting, affine-syndrome computation, and candidate generation cost $\widetilde O_{\eps,\ell_{\rm in}}(n)$.  For each candidate, forming $\widehat y_a$, running the affine outer decoder, and constructing and checking $e_a^\star$ cost $T_{\rm out}(n)+O(n)$.  This gives the stated running time.  The only randomized part of the algorithm is candidate generation, whose success probability is $1-o(1)$.
\end{proof}

\subsection{Candidate Enumeration via Weak Regularity}\label{subsec:csp-candidates}

We construct $\mathsf{CandGen}$ and prove \Cref{prop:candidate-generation}.  Throughout this subsection, the local lists and label maps are those defined in \Cref{subsec:syndrome-lifting}.  We first express local agreement as a CSP and identify the assignments associated with low-weight errors.  Weak regularity then gives a single factor on which all of these assignments are concentrated.

\paragraph{The agreement CSP.}
\begin{definition}\label{def:agreement-csp}
For each edge $(u,v)\in E$, with port $i=\mathrm{nbr}_u(v)$, let
\[
    \Psi_{(u,v)}:=\{t\in[\ell_{\rm in}]:
       \mathsf L_u(t)\ne\bot
       \text{ and }(\mathsf L_u(t))_i=(r_u)_i\}.
\]
Set $\Psi_{(u,v)}=\emptyset$ off $E$.  The CSP has one variable $A_u\in[\ell_{\rm in}]$ for every $u\in L$; an edge $(u,v)$ is satisfied when $A_u\in\Psi_{(u,v)}$.  Thus dummy labels satisfy no edges.  For each label $t$, define
\begin{equation}\label{eq:gt-definition}
    g_t(u,v):=\1\{(u,v)\in E\}\,\1\{t\in\Psi_{(u,v)}\}.
\end{equation}
For $S\subseteq L$, $T\subseteq R$, and $a:S\to[\ell_{\rm in}]$, write
\begin{equation}\label{eq:csp-value-definition}
    \val_{S,T}(a):=\sum_{t=1}^{\ell_{\rm in}}
       \left\langle g_t,
         \1_{\{u\in S:a(u)=t\}}\otimes\1_T\right\rangle.
\end{equation}
This is the number of satisfied edges in $E(S,T)$.
\end{definition}

\paragraph{Target assignments and local rigidity.}
The following sets depend on the unknown error and are used only in the analysis.  The decoder does not need to find them.

\begin{lemma}\label{lem:planted-completeness}
Let $H_Xe^R=s$ and $\wt_{\Delta b_{\rm in}}(e^R)\le\tau n$, and put $w=\Fold^{-1}(e^R)+r$.  Define
\[
    T_w:=\{v\in R:e^R_v=0\}
        =\{v\in R:\Fold(w)_v=r^R_v\},
\]
and
\[
    S_w:=\{u\in L:|N(u)\setminus T_w|\le\rho_{\rm in}\Delta\}.
\]
Then $|S_w|\ge(1-\alpha)n$, and there is an assignment $a_w:S_w\to[\ell_{\rm in}]$ with
\[
    \mathsf L_u(a_w(u))=w_u\quad(u\in S_w),
    \qquad
    \val_{S_w,T_w}(a_w)=|E(S_w,T_w)|.
\]
\end{lemma}

\begin{proof}
The weight bound gives $|T_w|\ge(1-\delta_R+\eps)n$.  Apply \Cref{lem:neighborhood-concentration} with $a=1-\delta_R+\eps$ and $\beta=\eps/2$.  All but at most
$4\lambda^2n/(\Delta^2\eps^2)=\alpha n$ left vertices have at least
$(1-\delta_R+\eps/2)\Delta$ neighbors in $T_w$, equivalently at most $\rho_{\rm in}\Delta$ neighbors outside $T_w$.  This gives the bound on $|S_w|$.

For $u\in S_w$, the words $w_u$ and $r_u$ differ only on edges leading to $R\setminus T_w$.  Thus
\[
    \mathrm{dist}_{b_{\rm in}}(w_u,r_u)\le\rho_{\rm in}\Delta.
\]
By \Cref{lem:local-lift-property}, $w_u\in\mcC_X$, so it occurs in $\mcL_u$ and has a unique genuine label $a_w(u)$.  On every edge of $E(S_w,T_w)$, the word $w$ agrees with $r$.  Therefore this assignment satisfies every such edge.
\end{proof}

\begin{lemma}[Local rigidity]\label{lem:local-robustness}
With $w,S_w,T_w,a_w$ as in \Cref{lem:planted-completeness}, fix $u\in S_w$ and $t\ne a_w(u)$.  Replacing the label at $u$ by $t$ violates at least $(\eps/2)\Delta$ edges of $E(\{u\},T_w)$.
\end{lemma}

\begin{proof}
If $\mathsf L_u(t)=\bot$, every edge from $u$ to $T_w$ is violated.  Since
\[
    |N(u)\cap T_w|\ge(1-\rho_{\rm in})\Delta\ge(\eps/2)\Delta,
\]
the claim follows.

Otherwise, put $c':=\mathsf L_u(t)$.  Distinctness of the genuine entries gives $c'\ne w_u$, so $h:=c'-w_u\in\mcC_X\setminus\{0\}$.  If $h\in\mcC_Z^\perp$, the stabilizer-distance hypothesis gives $\wt_{b_{\rm in}}(h)\ge\delta_R\Delta$; if $h\notin\mcC_Z^\perp$, the quotient-distance hypothesis gives the same bound.  At most $\rho_{\rm in}\Delta$ disagreement coordinates can lead outside $T_w$.  Hence at least
\[
    (\delta_R-\rho_{\rm in})\Delta=(\eps/2)\Delta
\]
disagreements lie on edges into $T_w$.  Since $w$ agrees with $r$ on those edges, each is violated by the new label.
\end{proof}

\begin{remark}\label{rem:quantum-collision}
The quotient-distance case separates locally different logical cosets.  The stabilizer-distance case separates distinct representatives of the same local coset, which may impose different edge constraints even though their projections under $\pi_X$ coincide.  This is why the rigidity argument uses both distance hypotheses and retains the genuine inner codewords in the local lists.
\end{remark}

\paragraph{Weak regularity and atomwise label counts.}
Apply \Cref{thm:weak-reg-expanded} to each $g_t$ with accuracy $\gamma$.  The spectral hypothesis \eqref{eq:spectral-condition-decoder} permits these calls.  On their joint success event, write
\[
    h_t=\sum_{j=1}^{p_t}c_{t,j}\,\1_{A_{t,j}}\otimes\1_{B_{t,j}},
    \qquad p_t=O(\gamma^{-2}),
    \qquad \|g_t-h_t\|_\square\le\gamma n\Delta.
\]
Let $\calB_L$ be the factor generated by all left cuts $A_{t,j}$.  Then
\begin{equation}\label{eq:left-factor-size}
    |\calB_L|\le2^{O(\ell_{\rm in}\gamma^{-2})}
              =2^{O(\ell_{\rm in}^3/\eps^6)}.
\end{equation}
For $a:S\to[\ell_{\rm in}]$, put $S_t(a):=\{u\in S:a(u)=t\}$.  Replacing each $g_t$ by $h_t$ in \eqref{eq:csp-value-definition} gives
\begin{equation}\label{eq:proxy-value}
    \widehat{\val}_{S,T}(a)
    :=\sum_{t=1}^{\ell_{\rm in}}\sum_{j=1}^{p_t}
       c_{t,j}\,|A_{t,j}\cap S_t(a)|\,|B_{t,j}\cap T|.
\end{equation}
The cut-norm bounds imply
\begin{equation}\label{eq:proxy-error}
    |\val_{S,T}(a)-\widehat{\val}_{S,T}(a)|
    \le\ell_{\rm in}\gamma n\Delta
\end{equation}
for every $S,T,a$.  Each $A_{t,j}$ is a union of atoms, so the proxy depends on $a$ only through its label counts in the restricted atoms $P\cap S$.

\begin{lemma}\label{lem:value-stability}
Let $S\subseteq L$, $T\subseteq R$, and let $a,a':S\to[\ell_{\rm in}]$ have the same label counts in every atom of $\calB_L|_S$.  Then
\[
    |\val_{S,T}(a)-\val_{S,T}(a')|\le2\ell_{\rm in}\gamma n\Delta.
\]
The same conclusion holds with $\ell$ in place of $\ell_{\rm in}$ for any agreement instance with $\ell$ labels and a factor generated by weak-regularity approximations at accuracy $\gamma$.
\end{lemma}

\begin{proof}
The atomwise counts determine every $|A_{t,j}\cap S_t(a)|$, so the two proxy values in \eqref{eq:proxy-value} are equal.  Applying \eqref{eq:proxy-error} to both assignments proves the bound.  The proof uses only the number of labels, giving the stated general form.
\end{proof}

\paragraph{Concentration on the factor.}
Value stability says that rearranging labels within each atom changes the value little while rigidity gives the opposite conclusion if many labels of a fully satisfying assignment change.

\begin{lemma}\label{lem:one-sided-concentration-expanded}
Let $g_1,\ldots,g_\ell$ define a left-oriented agreement CSP, and let $\calB_L$ be generated by weak-regularity approximations to these functions at accuracy $\gamma$.  Define $\val$ as in \eqref{eq:csp-value-definition}, with $\ell$ labels.  Suppose $\emptyset\ne S\subseteq L$, $T\subseteq R$, and $a:S\to[\ell]$ satisfy:
\begin{enumerate}[label=(\roman*),leftmargin=2.6em]
    \item $\val_{S,T}(a)=|E(S,T)|$;
    \item for every $u\in S$ and $t'\ne a(u)$, changing the label at $u$ to $t'$ violates at least $(\eps/2)\Delta$ edges into $T$.
\end{enumerate}
Then $a$ is $\eta$-concentrated on $\calB_L|_S$, where
\[
    \eta:=\frac{5\ell\gamma}{\eps}\cdot\frac{n}{|S|}.
\]
\end{lemma}

\begin{proof}
For $\ell=1$ the conclusion is immediate, so assume $\ell\ge2$.  Suppose, for contradiction, that $a$ is not $\eta$-concentrated.  Let $a^{\rm maj}$ assign a most frequent label in each restricted atom, breaking ties arbitrarily, and set
\[
    D:=\{u\in S:a(u)\ne a^{\rm maj}(u)\}.
\]
Since $a^{\rm maj}$ is $\calB_L|_S$-measurable, $|D|>\eta|S|$.

We construct an assignment $a'$ preserving every atomwise label count while changing every vertex in $D$.  Fix a restricted atom $P$, and order the labels $t_1,\ldots,t_\ell$ by their frequencies, with $t_1$ the chosen majority label.  Write
\[
    C_i:=\{u\in P:a(u)=t_i\},
    \qquad c_i:=|C_i|,
    \qquad c_1\ge\cdots\ge c_\ell,
\]
and set $c_{\ell+1}:=0$.  Since
\[
    \sum_{i=2}^{\ell}(c_i-c_{i+1})=c_2\le c_1,
\]
choose disjoint sets $M_2,\ldots,M_\ell\subseteq C_1$ with $|M_i|=c_i-c_{i+1}$, and define on $P$
\[
    a'(u):=
    \begin{cases}
      t_{i-1}, & u\in C_i,\ i\ge2,\\
      t_i, & u\in M_i,\ i\ge2,\\
      t_1, & u\in C_1\setminus\bigcup_{i\ge2}M_i.
    \end{cases}
\]
Every minority vertex changes label.  The new count of $t_i$ for $i\ge2$ is $c_{i+1}+|M_i|=c_i$, and the count of $t_1$ is $c_2+c_1-c_2=c_1$.  Carrying out this construction in each atom gives the required $a'$.

By \Cref{lem:value-stability},
\[
    \val_{S,T}(a)-\val_{S,T}(a')\le2\ell\gamma n\Delta.
\]
On the other hand, every vertex in $D$ changes label and loses at least $(\eps/2)\Delta$ satisfied edges by rigidity.  Edges incident to different left vertices are disjoint, and the original assignment satisfies every edge in $E(S,T)$, so there is no double counting or compensating gain.  Therefore
\begin{align*}
    \val_{S,T}(a)-\val_{S,T}(a')
    &\ge |D|\frac\eps2\Delta\\
    &>\eta|S|\frac\eps2\Delta
      =\frac52\ell\gamma n\Delta,
\end{align*}
a contradiction.
\end{proof}

\begin{proof}[Proof of \Cref{prop:candidate-generation}]
\emph{The procedure.}
From $r^R$, unfold to obtain $r$, compute the local lists and their label maps, and form the indicators $g_t$.  Run the weak-regularity calls above and construct $\calB_L$.  Enumerate every $\calB_L$-measurable assignment $a:L\to[\ell_{\rm in}]$ and compute its projected word $w_a$ using \eqref{eq:wa-definition}.  By \eqref{eq:left-factor-size}, this produces at most
\[
    \ell_{\rm in}^{|\calB_L|}
    \le\ell_{\rm in}^{\,2^{O(\ell_{\rm in}^3/\eps^6)}}
    =L_{\max}
\]
assignments.  The stated size and time bounds can be enforced by returning an empty family if a regularity call fails or exceeds its prescribed bounds.

\emph{Agreement.}
Condition on the joint success of the $\ell_{\rm in}$ regularity calls, and fix any target error $e^R$.  By \Cref{lem:planted-completeness,lem:local-robustness}, the associated assignment $a_w:S_w\to[\ell_{\rm in}]$ fully satisfies $E(S_w,T_w)$ and is rigid there.  The concentration lemma applies with
\[
    \eta_w:=\frac{5\ell_{\rm in}\gamma}{\eps}\frac{n}{|S_w|}
    \le\frac{5\eps^2}{2^5(1-\alpha)}
    =\eta_{\rm conc}.
\]
It gives a $\calB_L|_{S_w}$-measurable assignment differing from $a_w$ on at most
\[
    \eta_w|S_w|
    =\frac{5\ell_{\rm in}\gamma}{\eps}n
    =\frac{5\eps^2}{2^5}n
    \le\eta_{\rm conc}n
\]
vertices.  Extend it to all of $L$ by giving each full atom the label chosen on its intersection with $S_w$, and choosing any label for atoms disjoint from $S_w$.  The extension $a$ is among the enumerated assignments.

Whenever $u\in S_w$ keeps its label, $\mathsf L_u(a(u))=w_u$ and hence $w_a(u)=\pi_X(w_u)$.  Thus
\[
    \big|\{u\in L:w_a(u)\ne\pi_X(w_u)\}\big|
    \le |L\setminus S_w|+\eta_{\rm conc}n
    \le(\alpha+\eta_{\rm conc})n.
\]
The regularity success event provides cut-norm control for every rectangle.  It depends only on the fixed instance, not on $e^R,S_w$, or $T_w$.  Consequently this argument covers all target errors simultaneously, without a union bound over errors.

\emph{Probability and running time.}
There are only the constant number $\ell_{\rm in}$ of regularity calls, so their joint success probability is $1-o(1)$.  Local list construction and the CSP indicators take $O(n)$ time.  Weak regularity takes $\widetilde O_{\eps,\ell_{\rm in}}(n)$ time, and constructing the factor and writing all $L_{\max}$ candidates costs $O_{\eps,\ell_{\rm in}}(n)$.  This proves the claimed bound and completes the proof of \Cref{thm:syndrome-input-decoder}.
\end{proof}

\section{Explicit Instantiation}\label{sec:instantiation}
This section instantiates the abstract decoder of \Cref{thm:syndrome-input-decoder}.  We begin with the result and with a one-shot synthesis statement.  Only afterward do we introduce the outer code, the inner code, and the parameter choices.  This order separates the conceptual argument from the bookkeeping.

\begin{theorem*}[\Cref{thm:main}, restated]
For every target rate $R\in(0,1)$ and every $\zeta,\xi>0$, there are constants $b,\ell,w\in\bbN$, an infinite set $\mathcal N\subseteq\bbN$, and an explicit family $\{\mcQ_N:N\in\mathcal N\}$ of $\F_2$-linear vector-space CSS codes on $N$ folded blocks of size $b$ such that, for every $N\in\mathcal N$:
\begin{enumerate}[label=(\roman*),leftmargin=2.7em]
    \item the actual rate $R_N$ is at least $R$, and the binary row and column weights of both check matrices are at most $w$;
    \item an explicit certified bound $\delta_N\le\delta(\mcQ_N)$ satisfies
    \[
        \delta_N\ge\frac{1-R_N}{2}-\zeta;
    \]
    \item for some $\tau_N\ge\delta_N-\xi$, every $X$- or $Z$-syndrome has at most $\ell$ stabilizer cosets containing a representative of folded weight at most $\tau_NN$; and
    \item a randomized syndrome-input decoder finds a list containing all of those cosets in time $\widetilde O_{R,\zeta,\xi}(N)$ with probability $1-o(1)$, and verifies the syndrome of every output representative.
\end{enumerate}
\end{theorem*}

The following proposition gives the parameters of quantum AEL with an arbitrary inner and outer code. Based on these parameters, we will be able to instantiate our construction by picking the right inner and outer codes. 

\begin{proposition}\label{prop:one-shot-synthesis}
Let $\mcC$ be a fixed inner CSS code of length $\Delta$ and block size $b_{\rm in}$, and let $\mcD=\{\mcD_n\}_n$ be an outer CSS family over blocks of size
\[
    b_{\rm out}=K_2(\mcC).
\]
Assume the following data are available.
\begin{enumerate}[label=(\roman*),leftmargin=2.7em]
    \item The outer family has rate $R(\mcD_n)\ge R_{\rm out}$, relative block distance at least $\delta_{\rm out}>0$, bounded binary row and column weights, and a verified syndrome decoder correcting block weight at most $\rho_{\rm out}n$ up to stabilizers in time $T_{\rm out}(n)$.
    \item The inner spaces $\mcC_X$ and $\mcC_Z$ are classically $(\tau_{\rm in},\ell_{\rm in})$-list decodable, and their quotient and nonzero-stabilizer distances are at least $\tau_{\rm in}$ on both CSS sides.
    \item $G_n$ is an explicit $(n,\Delta,\lambda)$-bipartite expander family.
\end{enumerate}
Fix $\eps>0$ and define
\begin{equation}\label{eq:one-shot-certified}
    \delta_R:=\tau_{\rm in}-\frac{\lambda}{\Delta\delta_{\rm out}},
    \qquad
    \alpha:=\frac{4\lambda^2}{\Delta^2\eps^2},
    \qquad
    \eta_{\rm conc}:=\frac{5\eps^2}{2^5(1-\alpha)}.
\end{equation}
Suppose $0<\eps<\delta_R$ and
\begin{equation}\label{eq:one-shot-conditions}
    \frac{\lambda}{\Delta}
      <\frac{\eps^6}{2^{33}\ell_{\rm in}^2},
    \qquad
    \alpha+\eta_{\rm conc}\le\rho_{\rm out}.
\end{equation}
Then the qAEL family $\mcF_n=\mcF(\mcC,\mcD_n,G_n)$ has
\begin{align*}
    R(\mcF_n)&=R(\mcC)R(\mcD_n)\ge R(\mcC)R_{\rm out},\\
    \delta(\mcF_n)&\ge\delta_R,
\end{align*}
is LDPC, and is syndrome-input list decodable on both CSS sides up to radius $\delta_R-\eps$.  Its list size is at most
\[
    L_{\max}
    =\ell_{\rm in}^{\,2^{O(\ell_{\rm in}^3/\eps^6)}},
\]
and its running time is
\[
    \widetilde O_{\eps,\ell_{\rm in}}(n)
    +L_{\max}\bigl(T_{\rm out}(n)+O(n)\bigr).
\]
With probability $1-o_{n\to\infty}(1)$, the output contains every stabilizer coset having a representative of folded weight at most $(\delta_R-\eps)n$.  Every output representative has the prescribed qAEL syndrome.
\end{proposition}

\subsection{High-Rate Quantum Tanner Outer Codes}\label{subsec:qtc-outer}

We use quantum Tanner codes as a black-box outer code family.  The formulation below describes only the properties used by \Cref{prop:one-shot-synthesis}. The constants may deteriorate as the rate approaches one, but they nevertheless remain independent of the block length.  As in the main theorem, any constant-size local codes are found via brute force.

\begin{theorem}[Quantum Tanner Codes, \cite{LeverrierZemor22QuantumTanner,LeverrierZemor22DecodingQTanner}]\label{thm:qtc-blackbox}
For every rate loss $\eta_{\rm out}>0$, there are positive constants
\[
    \delta_{\rm out},\rho_{\rm out}>0
\]
and an explicit binary quantum Tanner CSS family $\mcD^{(1)}=\{\mcD^{(1)}_n\}_n$ with rate at least $1-\eta_{\rm out}$ and relative distance at least $\delta_{\rm out}$.  The family is LDPC and has verified $X$- and $Z$-syndrome decoders which, for every error of weight at most $\rho_{\rm out}n$, return a correction in the correct stabilizer class in time $\widetilde O_{\eta_{\rm out}}(n)$.
\end{theorem}

The qAEL outer alphabet is $\F_2^{b_{\rm out}}$ rather than $\F_2$.  We therefore take the direct sum of $b_{\rm out}$ copies of $\mcD^{(1)}$.  This operation preserves the folded rate.  It also preserves the relative distance and decoding radius in block Hamming weight: an error supported on at most $\rho_{\rm out}n$ blocks induces, in each binary coordinate, an error supported on at most the same number of positions.  Running the binary decoder independently in every coordinate costs $\widetilde O(b_{\rm out}n)$, which is near-linear because $b_{\rm out}$ is fixed.  Finally, \Cref{lem:affine-outer-decoder} converts these syndrome decoders into the affine outer interface used in the stitching stage.

\subsection{A Constant-Size Inner Code at the Singleton Bound}\label{subsec:inner-instantiation}

The inner code is found once via brute force.  Its block length is the expander degree $\Delta$, which is a constant independent of the block length, so exhaustive decoding and exhaustive verification can be done in constant time.

For a prime power $Q$, let
\[
    H_Q(x)
    :=x\log_Q(Q-1)-x\log_Qx-(1-x)\log_Q(1-x)
\]
denote the $Q$-ary entropy function.

\begin{theorem}\label{thm:random-inner-package}
Fix a rational inner rate $R_{\rm in}\in(0,1)$ and a slack $\eps_{\rm in}\in(0,1-R_{\rm in})$, and set
\[
    \tau_{\rm in}:=\frac{1-R_{\rm in}-\eps_{\rm in}}2.
\]
Choose $Q=2^{b_{\rm in}}$ so that
\begin{equation}\label{eq:entropy-inner-condition}
    H_Q(\tau_{\rm in})
      \le \tau_{\rm in}+\frac{\eps_{\rm in}}4.
\end{equation}
For every sufficiently large admissible constant $\Delta$, there exists an $\F_2$-linear vector-space CSS code
\[
    \mcC=(\mcC_X,\mcC_Z)
    \subseteq(\F_2^{b_{\rm in}})^\Delta
\]
of rate $R_{\rm in}$ such that
\begin{align*}
    \delta(\mcC)&\ge\tau_{\rm in},\\
    \delta_{\rm quot}(\mcC_X),\ \delta_{\rm stab}(\mcC_X)&\ge\tau_{\rm in},\\
    \delta_{\rm quot}(\mcC_Z),\ \delta_{\rm stab}(\mcC_Z)&\ge\tau_{\rm in}.
\end{align*}
Moreover, both $\mcC_X$ and $\mcC_Z$ are classically list decodable at block radius $\tau_{\rm in}\Delta$ with a list bound
\[
    \ell_{\rm in}=O_{\tau_{\rm in},Q}(1/\eps_{\rm in}).
\]
A code with these properties can be found by brute force.
\end{theorem}

A proof is given in \Cref{app:inner-existence}.  The construction samples an orthogonal pair of $\F_Q$-subspaces, then uses trace-dual bases to consider the resulting CSS code as $\F_2$-linear without changing block support.  Uniform-marginal and first-moment arguments give the four distance requirements, while the random-linear-code theorem of Guruswami--H\aa stad--Kopparty~\cite{DBLP:journals/corr/abs-1001-1386} gives the constant local lists.  Notice that the separate stabilizer-distance clause is essential for the quantum rigidity lemma in \Cref{sec:algorithm}; ordinary quotient distance alone would not exclude two local hypotheses differing by a low-weight stabilizer.

Since
\[
    H_Q(x)\le x+\frac{H_2(x)}{\log_2Q},
\]
it suffices to take
\[
    b_{\rm in}\ge\frac{4H_2(\tau_{\rm in})}{\eps_{\rm in}}.
\]
Thus, the block alphabet remains constant once the target parameters are fixed.

\subsection{Expander and Stitching Parameters}\label{subsec:parameter-synthesis}

Fix a small expander slack $\zeta_{\rm exp}>0$ and write
\[
    A_{\rm exp}:=2+\zeta_{\rm exp}.
\]
By \Cref{thm:alon-prescribed-size}, after taking bipartite double covers, for every sufficiently large admissible $n$ there is an explicit $(n,\Delta,\lambda)$-bipartite expander with
\begin{equation}\label{eq:near-ramanujan-used}
    \frac{\lambda}{\Delta}
      \le\frac{A_{\rm exp}}{\sqrt\Delta}.
\end{equation}
Consequently the qAEL distance theorem certifies
\begin{equation}\label{eq:certified-distance-instantiation}
    \delta_R
    :=\tau_{\rm in}
      -\frac{A_{\rm exp}}{\delta_{\rm out}\sqrt\Delta}
    \le\delta(\mcF).
\end{equation}

The following elementary observation reduces the stitching condition to two inequalities.

\begin{lemma}\label{lem:stitching-feasibility}
If $\alpha\le1/2$, then
\[
    \eta_{\rm conc}
      =\frac{5\eps^2}{2^5(1-\alpha)}
      \le\frac{5\eps^2}{16}.
\]
Hence \eqref{eq:stitching-inequality} follows whenever
\[
    \alpha\le\frac{\rho_{\rm out}}2
    \qquad\text{and}\qquad
    \frac{5\eps^2}{16}\le\frac{\rho_{\rm out}}2.
\]
\end{lemma}

\begin{proof}
The first conclusion is immediate from $(1-\alpha)^{-1}\le2$.  The second follows by adding the two displayed upper bounds.
\end{proof}

For fixed $\eps,\ell_{\rm in},\delta_{\rm out},\rho_{\rm out}>0$, all remaining requirements can therefore be met by taking the constant degree $\Delta$ sufficiently large.  Indeed, \eqref{eq:near-ramanujan-used} makes the spectral ratio $O(\Delta^{-1/2})$, the distance loss $O(\Delta^{-1/2})$, and $\alpha=O(\Delta^{-1})$.

\subsection{Proof of the Quantum AEL Parameters}\label{subsec:one-shot-proof}

\begin{proof}[Proof of \Cref{prop:one-shot-synthesis}]
By \Cref{prop:qael-rate},
\[
    R(\mcF_n)=R(\mcC)R(\mcD_n).
\]
By \Cref{thm:distance-amplification} and the assumed inner and outer distance bounds,
\[
    \delta(\mcF_n)
      \ge\delta(\mcC)-\frac{\lambda}{\Delta\delta(\mcD_n)}
      \ge\tau_{\rm in}-\frac{\lambda}{\Delta\delta_{\rm out}}
      =\delta_R.
\]
Since $\delta_R\le\tau_{\rm in}$, the inner list-decoding, quotient-distance, and stabilizer-distance assumptions remain valid at the thresholds $\rho_{\rm in}=\delta_R-\eps/2$ and $\delta_R$ used in \Cref{thm:syndrome-input-decoder}.  Conditions \eqref{eq:one-shot-conditions} are exactly that theorem's spectral and stitching hypotheses.  Applying it on the $X$ side, and then symmetrically on the $Z$ side, gives the radius, list-size, runtime, completeness, and syndrome-verification conclusions.  Finally, \Cref{prop:qael-ldpc} gives the LDPC property because the inner parameters are fixed and the outer family is LDPC.
\end{proof}

\subsection{Proof of the Main Theorem}\label{subsec:main-instantiation-proof}

\begin{proof}[Proof of \Cref{thm:main}]
Fix $R\in(0,1)$ and $\zeta,\xi>0$.  Choose a rational $R_{\rm in}$ such that 
\[
    R<R_{\rm in}<1.
\]
Next choose an outer rate loss $\eta_{\rm out}>0$ small enough that
\begin{equation}\label{eq:outer-rate-budget}
    R_{\rm in}(1-\eta_{\rm out})\ge R,
    \qquad
    \frac{R_{\rm in}\eta_{\rm out}}2\le\frac\zeta3.
\end{equation}
Apply \Cref{thm:qtc-blackbox} with this $\eta_{\rm out}$, obtaining constants $\delta_{\rm out},\rho_{\rm out}>0$ and an outer family whose actual rates $R_{{\rm out},N}$ satisfy $R_{{\rm out},N}\ge1-\eta_{\rm out}$.

Choose $\eps_{\rm in}>0$ so that
\begin{equation}\label{eq:inner-gap-budget}
    \eps_{\rm in}<1-R_{\rm in},
    \qquad
    \frac{\eps_{\rm in}}2\le\frac\zeta3,
\end{equation}
and set $\tau_{\rm in}=(1-R_{\rm in}-\eps_{\rm in})/2$.  Choose $b_{\rm in}$ satisfying \eqref{eq:entropy-inner-condition}.  Finally choose the decoding slack $\eps>0$ such that
\begin{equation}\label{eq:decoder-slack-budget}
    \eps<\min\left\{
      \xi,\frac{\tau_{\rm in}}4,
      \sqrt{\frac{8\rho_{\rm out}}5}
    \right\}.
\end{equation}

We now choose one sufficiently large admissible constant $\Delta$.  In addition to the integrality and existence requirements of \Cref{thm:random-inner-package}, require
\begin{align}
    \frac{A_{\rm exp}}{\sqrt\Delta}
      &<\frac{\eps^6}{2^{33}\ell_{\rm in}^2},
      \label{eq:final-spectral-budget}\\
    \frac{4A_{\rm exp}^2}{\Delta\eps^2}
      &\le\min\left\{\frac12,\frac{\rho_{\rm out}}2\right\},
      \label{eq:final-alpha-budget}\\
    \frac{A_{\rm exp}}{\delta_{\rm out}\sqrt\Delta}
      &\le\min\left\{\frac\zeta3,\frac{\tau_{\rm in}}2\right\}.
      \label{eq:final-distance-budget}
\end{align}
Such a constant exists because every left-hand side tends to zero with $\Delta$.  Fix and verify an inner code supplied by \Cref{thm:random-inner-package}, set
\[
    b_{\rm out}=K_2(\mcC)=b_{\rm in}R_{\rm in}\Delta,
\]
and take the $b_{\rm out}$-fold direct sum of the outer quantum Tanner family.  Let $G_N$ be the explicit expander from \Cref{thm:alon-prescribed-size}, and form the qAEL code $\mcQ_N$.

Equations \eqref{eq:decoder-slack-budget}--\eqref{eq:final-alpha-budget} and \Cref{lem:stitching-feasibility} verify the stitching inequality; \eqref{eq:final-spectral-budget} verifies the spectral condition.  Hence \Cref{prop:one-shot-synthesis} applies with
\[
    \delta_N:=\tau_{\rm in}
      -\frac{A_{\rm exp}}{\delta_{\rm out}\sqrt\Delta},
    \qquad
    \tau_N:=\delta_N-\eps.
\]
The final actual rate is
\[
    R_N=R_{\rm in}R_{{\rm out},N}
      \ge R_{\rm in}(1-\eta_{\rm out})
      \ge R.
\]
Moreover,
\begin{align*}
    \frac{1-R_N}{2}-\delta_N
      &=\frac{R_{\rm in}(1-R_{{\rm out},N})}{2}
        +\frac{\eps_{\rm in}}2
        +\frac{A_{\rm exp}}{\delta_{\rm out}\sqrt\Delta}\\
      &\le\frac{R_{\rm in}\eta_{\rm out}}2
        +\frac{\eps_{\rm in}}2
        +\frac{A_{\rm exp}}{\delta_{\rm out}\sqrt\Delta}\\
      &\le\zeta.
\end{align*}
Thus $\delta_N\ge(1-R_N)/2-\zeta$, and
\[
    \tau_N=\delta_N-\eps\ge\delta_N-\xi.
\]
The fold size $b=\Delta b_{\rm in}$, the list bound
\[
    \ell=\ell_{\rm in}^{\,2^{O(\ell_{\rm in}^3/\eps^6)}},
\]
and the binary LDPC weight bound are constants depending only on $(R,\zeta,\xi)$.  Since the outer decoder is near-linear and all remaining work is near-linear in $N$, \Cref{prop:one-shot-synthesis} gives the claimed $\widetilde O_{R,\zeta,\xi}(N)$ syndrome-input decoders on both CSS sides.
\end{proof}
 
\section*{Acknowledgments}
The research direction and conceptual aspects of this work are all human. 
The authors used all versions of ChatGPT 5 Pro for research exploration (mainly as
an accelerator) and for editorial assistance in reorganizing and polishing the manuscript. 
We used Claude Opus 4.6 for an extensive Lean formalization~\footnote{\url{https://github.com/Granha/quantum_ael_lean_formalization}} 
of an earlier version of this paper. The authors take responsibility 
for the mathematical claims, proofs, and citations.

\newcommand{\etalchar}[1]{$^{#1}$}

\appendix
\section{Deferred Proofs from the Preliminaries}\label{app:prelim-proofs}

This appendix contains the short expander and linear-algebra arguments deferred from \Cref{sec:preliminaries}.  Keeping them here lets the preliminary section function as a compact reference while making the manuscript self-contained.

\begin{proof}[Proof of \Cref{lem:neighborhood-concentration}]
Let $B:=L\setminus S$.  Every $u\in B$ has fewer than $(a-\beta)\Delta$ neighbors in $T$, and hence
\[
    |E(B,T)|<(a-\beta)\Delta|B|.
\]
On the other hand, the expander mixing lemma and $|T|\ge an$ give
\begin{align*}
    |E(B,T)|
      &\ge \frac{\Delta}{n}|B||T|
        -\lambda\sqrt{|B||T|}\\
      &\ge a\Delta|B|-\lambda\sqrt{n|B|}.
\end{align*}
If $B\ne\emptyset$, comparison yields
\[
    \beta\Delta|B|<\lambda\sqrt{n|B|},
\]
so $|B|<\lambda^2n/(\Delta^2\beta^2)$.  The conclusion follows from $|S|=n-|B|$.  Interchanging the two sides of the bipartite graph proves the symmetric statement.
\end{proof}

\begin{proof}[Proof of \Cref{lem:duality-embeddings}]
The bilinear pairing between $\mcW_X$ and $\mcW_Z$ is nondegenerate.  Indeed, suppose $x\in\mcW_X$ is orthogonal to $\mcW_Z$.  Since $x\in\mcC_X$, it is also orthogonal to $\mcC_X^\perp$.  The decomposition
\[
    \mcC_Z=\mcW_Z\oplus\mcC_X^\perp
\]
therefore implies $x\in\mcC_Z^\perp$.  But
$\mcW_X\cap\mcC_Z^\perp=\{0\}$ by \eqref{eq:WXWZ}, so $x=0$.  The same argument with $X$ and $Z$ interchanged proves nondegeneracy in the second argument.

Both spaces have dimension
\[
    \dim\mcW_X
      =\dim\mcC_X-\dim\mcC_Z^\perp
      =b_{\rm out}
      =\dim\mcW_Z.
\]
Choose bases $x_1,\ldots,x_{b_{\rm out}}$ of $\mcW_X$ and $z_1,\ldots,z_{b_{\rm out}}$ of $\mcW_Z$, and let $M$ be the matrix with entries $M_{ij}=\langle x_i,z_j\rangle$.  Nondegeneracy makes $M$ invertible.  Define
\[
    \phi_X(a):=\sum_i a_ix_i,
    \qquad
    \phi_Z(b):=\sum_j (M^{-1}b)_jz_j.
\]
Then $\langle\phi_X(a),\phi_Z(b)\rangle=a^\mathsf{T}b$.  Computing $M^{-1}$ is Gaussian elimination over $\F_q$.
\end{proof}

\begin{proof}[Proof of \Cref{prop:qael-rate}]
By \Cref{prop:dual-decomposition}, before folding we have
\[
    \mcE_Z^\perp
      =\Phi_X(\mcD_Z^\perp)+(\mcC_Z^\perp)^L
      \subseteq
      \mcE_X
      =\Phi_X(\mcD_X)+(\mcC_Z^\perp)^L.
\]
The image of $\Phi_X$ lies in $\mcW_X^L$, which intersects $(\mcC_Z^\perp)^L$ trivially.  Consequently the map induced by $\Phi_X$ gives an isomorphism
\[
    \mcD_X/\mcD_Z^\perp
      \cong
    \mcE_X/\mcE_Z^\perp.
\]
Folding is a coordinate permutation, so it preserves this quotient dimension.  Hence
\[
    K_q(\mcF)=K_q(\mcD).
\]
Since $b_{\rm out}=K_q(\mcC)$, the folded block size of $\mcF$ is $\Delta b_{\rm in}$, and its blocklength is $n$, we obtain
\begin{align*}
    R(\mcF)
      &=\frac{K_q(\mcD)}{n\Delta b_{\rm in}}\\
      &=\frac{K_q(\mcC)}{\Delta b_{\rm in}}
        \cdot
        \frac{K_q(\mcD)}{nb_{\rm out}}
       =R(\mcC)R(\mcD).
\end{align*}
\end{proof}

\begin{proof}[Proof of \Cref{prop:qael-ldpc}]
Use the spanning checks in \eqref{eq:qael-x-checks}, and the symmetric spanning set on the $Z$ side.  A local inner row is supported on one left inner block, so its binary weight is at most $\Delta b_{\rm in}$, a constant.  If an outer row has binary weight at most $w_r^{\rm out}$, its image under the fixed coordinatewise map $\Phi_Z$ has binary weight at most $w_r^{\rm out}\Delta b_{\rm in}$, also a constant.  Folding does not change row weight.

For column weight, a physical binary coordinate belongs to one inner block.  It appears in at most $\dim\mcC_X^\perp$ local $X$-checks.  Among lifted outer checks, it can occur only in rows touching one of the $b_{\rm out}$ binary outer coordinates at the corresponding left vertex; if the outer column weight is at most $w_c^{\rm out}$, there are at most $b_{\rm out}w_c^{\rm out}$ such rows.  All quantities are constant.  The same argument applies to $Z$-checks, proving bounded row and column weights for the qAEL family.
\end{proof}

\section{Existence of the Constant-Size Inner Code}\label{app:inner-existence}

We prove \Cref{thm:random-inner-package}.  The proof uses two elementary transfer facts and the list-decodability theorem for random linear codes.  Throughout this section, $Q$ is a prime power and orthogonality in $\F_Q^\Delta$ is with respect to the standard dot product.

\begin{definition}\label{def:orthogonal-pair}
For $r\le\Delta/2$, sample uniformly from all pairs $(A,B)$ of $r$-dimensional subspaces of $\F_Q^\Delta$ satisfying $A\perp B$.
\end{definition}

\begin{lemma}\label{lem:orthogonal-uniform-marginals}
In the orthogonal-pair ensemble, $A$ and $B$ are each uniformly distributed among the $r$-dimensional subspaces of $\F_Q^\Delta$.  Their orthogonal complements are uniformly distributed among the $(\Delta-r)$-dimensional subspaces.
\end{lemma}

\begin{proof}
For a fixed $r$-dimensional $A$, the admissible choices of $B$ are exactly the $r$-dimensional subspaces of the $(\Delta-r)$-dimensional space $A^\perp$.  Their number depends only on the two dimensions, not on $A$.  Hence every $A$ occurs in the same number of admissible pairs.  Symmetry gives the claim for $B$, and orthogonal complementation is a bijection between the two relevant Grassmannians.
\end{proof}

\begin{lemma}\label{lem:flag-coupling}
Let $k'\le k\le\Delta$.  Sample a uniform $k$-dimensional subspace $V\subseteq\F_Q^\Delta$ and then a uniform $k'$-dimensional subspace $U\subseteq V$.  The marginal distribution of $U$ is uniform among all $k'$-dimensional subspaces of $\F_Q^\Delta$.  Consequently, if a random $k$-dimensional linear code is $(\rho,L)$-list decodable with probability at least $1-p$, then the same is true for a random $k'$-dimensional linear code.
\end{lemma}

\begin{proof}
For a fixed $U$, the number of $k$-dimensional superspaces $V$ containing it depends only on $(\Delta,k,k')$, so the marginal is uniform.  List decodability is monotone under passage to a subcode: every Hamming ball contains no more words of $U$ than of $V$.  The probabilistic conclusion follows under the stated coupling.
\end{proof}

We use the following form of the random-linear-code theorem.

\begin{theorem}[\cite{DBLP:journals/corr/abs-1001-1386}]\label{thm:ghk-random-linear}
For every prime power $Q$, every $p\in(0,1-1/Q)$, and every fixed $\eta>0$, a uniformly random $\F_Q$-linear code of rate $1-H_Q(p)-\eta$ is $(p,L)$-list decodable with probability $1-Q^{-\Omega(\Delta)}$, where
\[
    L=O_{p,Q}(1/\eta).
\]
\end{theorem}

We also use the standard point probability
\begin{equation}\label{eq:point-in-random-subspace}
    \Pr[x\in U]
      =\frac{Q^k-1}{Q^\Delta-1}
      \le Q^{k-\Delta}
\end{equation}
for a fixed nonzero $x\in\F_Q^\Delta$ and a uniform $k$-dimensional subspace $U$.

\begin{proof}[Proof of \Cref{thm:random-inner-package}]
Let
\[
    k:=R_{\rm in}\Delta,
    \qquad
    r:=\frac{\Delta-k}{2},
    \qquad
    \rho:=\frac{\Delta-r}{\Delta}=\frac{1+R_{\rm in}}2,
\]
which are integral for admissible $\Delta$.  Sample $(A,B)$ from the orthogonal-pair ensemble with $\dim A=\dim B=r$.

\paragraph{Restriction of scalars.}
Choose trace-dual $\F_2$-bases
$\alpha_1,\ldots,\alpha_{b_{\rm in}}$ and
$\beta_1,\ldots,\beta_{b_{\rm in}}$ of $\F_Q$, so
\[
    \Tr_{\F_Q/\F_2}(\alpha_i\beta_j)=\delta_{ij}.
\]
Let $\iota_X$ expand $\F_Q$-coordinates in the $\alpha$-basis and let $\iota_Z$ expand them in the $\beta$-basis.  For all $x,z\in\F_Q^\Delta$,
\[
    \langle\iota_X(x),\iota_Z(z)\rangle_{\F_2}
      =\Tr_{\F_Q/\F_2}\!\left(\langle x,z\rangle_{\F_Q}\right).
\]
It follows, by nondegeneracy of the trace and a dimension count, that for every $\F_Q$-subspace $V$,
\begin{equation}\label{eq:trace-dual-transfer}
    \iota_X(V)^\perp=\iota_Z(V^\perp),
    \qquad
    \iota_Z(V)^\perp=\iota_X(V^\perp).
\end{equation}
Both maps preserve block support.

Define
\[
    \mcC_Z^\perp:=\iota_X(A),
    \qquad
    \mcC_X:=\iota_X(B^\perp),
    \qquad
    \mcC_X^\perp:=\iota_Z(B),
    \qquad
    \mcC_Z:=\iota_Z(A^\perp).
\]
Equation \eqref{eq:trace-dual-transfer} makes these assignments consistent.  Moreover, $A\perp B$ gives $A\subseteq B^\perp$, hence $\mcC_Z^\perp\subseteq\mcC_X$.  The logical binary dimension is
\[
    b_{\rm in}\bigl((\Delta-r)-r\bigr)
      =b_{\rm in}k,
\]
so the folded rate is $k/\Delta=R_{\rm in}$.

\paragraph{Distance events.}
Let $U$ be a uniform $k'$-dimensional subspace of $\F_Q^\Delta$.  By \eqref{eq:point-in-random-subspace} and the $Q$-ary Hamming-ball estimate,
\begin{equation}\label{eq:random-subspace-distance}
    \Pr\bigl[d(U)\le\tau_{\rm in}\Delta\bigr]
      \le
      Q^{\Delta H_Q(\tau_{\rm in})+k'-\Delta}.
\end{equation}
For $k'=\Delta-r=\rho\Delta$, the exponent divided by $\Delta$ is at most
\[
    \tau_{\rm in}+\frac{\eps_{\rm in}}4+\rho-1
      =-\frac{\eps_{\rm in}}4.
\]
For $k'=r=(1-R_{\rm in})\Delta/2$, it is at most
\[
    \tau_{\rm in}+\frac{\eps_{\rm in}}4
      +\frac{1-R_{\rm in}}2-1
      =-R_{\rm in}-\frac{\eps_{\rm in}}4.
\]
By \Cref{lem:orthogonal-uniform-marginals} and a union bound, with probability $1-Q^{-\Omega(\Delta)}$ all four spaces
$A,B,A^\perp,B^\perp$ have distance greater than $\tau_{\rm in}\Delta$.  Because $\iota_X$ and $\iota_Z$ preserve block support, this simultaneously gives
\begin{itemize}[leftmargin=2em]
    \item nonzero $X$- and $Z$-stabilizer weight greater than $\tau_{\rm in}\Delta$;
    \item $d(\mcC_X),d(\mcC_Z)>\tau_{\rm in}\Delta$, and hence both quotient distances greater than $\tau_{\rm in}$; and
    \item $\delta(\mcC)>\tau_{\rm in}$.
\end{itemize}

\paragraph{Local list bounds.}
The entropy condition rearranges to
\[
    1-H_Q(\tau_{\rm in})-\frac{\eps_{\rm in}}4
      \ge
      1-\tau_{\rm in}-\frac{\eps_{\rm in}}2
      =\rho.
\]
Apply \Cref{thm:ghk-random-linear} at radius $\tau_{\rm in}$ and fixed rate slack $\eps_{\rm in}/4$.  It gives list size
\[
    \ell_{\rm in}=O_{\tau_{\rm in},Q}(1/\eps_{\rm in})
\]
for a uniform subspace of dimension
$\lfloor(1-H_Q(\tau_{\rm in})-\eps_{\rm in}/4)\Delta\rfloor$ with failure probability $Q^{-\Omega(\Delta)}$.  This dimension is at least $\rho\Delta$ for admissible sufficiently large $\Delta$.  By \Cref{lem:flag-coupling}, a uniform $(\Delta-r)$-dimensional subspace has the same list-decoding guarantee.  The uniform-marginal lemma therefore applies it to both $B^\perp$ and $A^\perp$, and the support-preserving maps transfer the result to $\mcC_X$ and $\mcC_Z$ in the folded block metric.

A final union bound shows that all distance and list-decoding properties hold simultaneously with positive probability.  Since $Q$ and $\Delta$ are constants of the eventual family, one may search for a successful pair $(A,B)$ offline, verify every required property exhaustively, and hardwire the resulting inner code and brute-force local decoder.
\end{proof}

\end{document}